\documentclass[acmsmall,screen,
authorversion=true,
]{acmart} 
\setcopyright{rightsretained}
\acmPrice{}
\acmDOI{10.1145/3371099}
\acmYear{2020}
\copyrightyear{2020}
\acmJournal{PACMPL}
\acmVolume{4}
\acmNumber{POPL}
\acmMonth{1}

\usepackage[utf8]{inputenc}
\usepackage{verbatim}
\usepackage[T1]{fontenc}
\usepackage{microtype}

\usepackage{mathrsfs}

\usepackage{bussproofs}
\usepackage{tikz}
\usepackage{tikz-cd}
\usepackage{csquotes}
\usepackage{xcolor}

\usepackage[all,2cell]{xy}\SelectTips{cm}{10}\SilentMatrices\UseAllTwocells  
\usepackage{mathtools}
\mathtoolsset{centercolon}

\usepackage{hyperref}
\hypersetup{breaklinks=true}

\newcommand{\Id}{\ensuremath{\mathsf{Id}}}

\newcommand{\pullback}{\ar@{}[dr]|<<{\lrcorner}}
\newcommand{\pushout}{\ar@{}[ul]|<<{\ulcorner}}

\newcommand{\refl}{\ensuremath{\mathsf{refl}}}
\newcommand{\trans}{\ensuremath{\mathsf{trans}}}
\newcommand{\appcongr}{\ensuremath{\mathsf{app\text{-}cong}}}
\newcommand{\abscongr}{\ensuremath{\mathsf{abs\text{-}cong}}}

\newcommand{\monsubst}[1]{\{#1\}}
\newcommand{\monadfromred}[1]{\underline{#1}}
\newcommand{\Mon}{\ensuremath{\mathsf{Mon}}}
\newcommand{\GMon}{\ensuremath{\mathsf{RedMon}}}
\newcommand{\Mod}{\ensuremath{\mathsf{Mod}}}

\newcommand{\LMod}{\ensuremath{\int_R\Mod(R)}}

\DeclareMathOperator{\codom}{codom}
\DeclareMathOperator{\colim}{colim}

\newcommand{\LC}{{\mathsf{LC}}}

\newcommand{\LCcong}{\ensuremath{\LC_{1\text{-cong}}}}
\newcommand{\LCclot}{\ensuremath{\LCb^*}}
\newcommand{\LCwhb}{\ensuremath{\LC_{wh\beta}}}

\newcommand{\LCb}{\ensuremath{\LC_{\beta}}}
\newcommand{\LCbpar}{\ensuremath{\LC_{\beta\parallel}}}

\newcommand\LCex{{\LC_{\mathsf{ex}}}}

\newcommand\LCfix{{\LC_{\fix}}}
\newcommand{\var}{{\mathsf{var}}}
\newcommand{\app}{{\mathsf{app}}}
\newcommand{\abs}{{\mathsf{abs}}}

\newcommand{\fix}{{\mathsf{fix}}}

\newcommand{\esubst}{{\mathsf{esubst}}}
\newcommand{\Cat}[1]{\mathsf{#1}}

\newcommand{\CC}{\Cat{C}}

\newcommand{\Set}{{\Cat{Set}}}

\newcommand{\swap}{{\mathsf{swap}}}

\newcommand{\Red}{\ensuremath{\mathsf{Red}}}
\newcommand{\red}{\ensuremath{\mathsf{red}}}
\newcommand{\Redof}[1]{\Red(#1)}
\newcommand{\redof}[1]{\red_{#1}}

\newcommand{\id}{\ensuremath{\mathsf{id}}}

\newcommand{\Hyp}{\ensuremath{\mathsf{Hyp}}}
\newcommand{\hyp}{\ensuremath{\mathsf{hyp}}}
\newcommand{\MVar}{\ensuremath{\mathsf{MVar}}}
\newcommand{\Concl}{{\mathsf{Con}}}
\newcommand{\concl}{{\mathsf{con}}}

\newcommand{\phyp}{h}
\newcommand{\pcon}{c}

\newcommand{\fv}{\ensuremath{\mathsf{fv}}}

\newcommand{\et}{\ensuremath{\eta}}

\newcommand{\redrule}[2]{\ensuremath{{#1}
\rightsquigarrow {#2}}}
\newcommand{\redfib}[2]{\ensuremath{{#1}
\blacktriangleright {#2}}}

\newcommand{\source}[1]{\ensuremath{\mathsf{source}_{#1}}}
\newcommand{\target}[1]{\ensuremath{\mathsf{target}_{#1}}}

\newcommand{\s}{\ensuremath{\mathsf{esubst}}}

\newcommand{\dswap}[1]{#1^{\vee}}

\newcommand{\MMM}{Metavariables}
\newcommand{\HHH}{Hypotheses}
\newcommand{\CCC}{Conclusion}

\newcommand{\N}{\ensuremath{\mathbb{N}}}

\newcounter{nameOfYourChoice}

\begin{document}
\title{Reduction Monads and Their Signatures}

\author{Benedikt Ahrens}
\orcid{0000-0002-6786-4538}
\affiliation{
  \department{School of Computer Science}
  \institution{University of Birmingham}
  \country{United Kingdom}
}
\email{B.Ahrens@cs.bham.ac.uk}

\author{Andr\'e Hirschowitz}
\orcid{0000-0003-2523-1481}
\affiliation{
  \department{LJAD}
  \institution{Université Côte d'Azur, CNRS}
  \country{France}
}
\email{ah@unice.fr}

\author{Ambroise Lafont}
\orcid{0000-0002-9299-641X}
\affiliation{
  \department{Gallinette, Inria}
  \institution{IMT Atlantique}
  \country{France}
}
\email{ambroise.lafont@inria.fr}

\author{Marco Maggesi}
\orcid{0000-0003-4380-7691}
\affiliation{
  \department{DiMaI}
  \institution{Università degli Studi di Firenze}
  \country{Italy}
}
\email{marco.maggesi@unifi.it}

\theoremstyle{acmdefinition}
\newtheorem{rem}[theorem]{Remark}
\newtheorem{notation}[theorem]{Notation}

\begin{abstract}

In this work, we study \emph{reduction monads}, which are essentially the same as monads relative to the free functor from sets into multigraphs.
Reduction monads account for two aspects of the lambda calculus: 
on the one hand, in the monadic viewpoint, the lambda calculus is an object equipped with a well-behaved substitution; 
on the other hand, in the graphical viewpoint, it is an oriented multigraph whose vertices are terms and whose edges witness the reductions between two terms.  

We study presentations of reduction monads. 
To this end, we propose a notion of \emph{reduction signature}. 
As usual, such a signature plays the role of a virtual presentation,
and specifies arities for generating operations---possibly subject to equations---together with arities for generating reduction rules. 
For each such signature, we define a category of models; 
any model is, in particular, a reduction monad. 
If the initial object of this category of models exists, we call it the \emph{reduction monad presented (or specified) by the given reduction signature}.
 
Our main result identifies a class of reduction signatures which specify a reduction monad in the above sense.
We show in the examples that our approach covers several standard variants of the lambda calculus.

\end{abstract}



\maketitle


\section{Introduction}
The lambda calculus has been a central object in theoretical computer science for decades. However, the corresponding mathematical structure does not seem to have
been identified once and for all. In particular, two complementary viewpoints on the (pure untyped)
lambda calculus have been widespread: some consider it as a
multigraph (or a relation, or a preorder, or even a category), while others view it as a monad (on the category of sets). The first account incorporates the
$\beta$-reduction, while the second addresses substitution but incorporates
only the $\beta$-equality. 
Merging these two perspectives led Lüth and Ghani \cite{DBLP:conf/ctcs/LuethG97} to consider monads on the category of preordered sets, and Ahrens \cite{ahrens_relmonads} to consider monads relative to the free functor from sets into preorders. 
In the present work, we propose a variant of their approaches. 
Here we call \emph{reduction monad} a monad relative to the discrete injection of sets in multigraphs, and of course the lambda calculus yields such a reduction monad. 
Our main contribution concerns the generation of reduction monads by syntactic (possibly binding) operations (possibly subject to equations) and reduction rules.
As is common in similar contexts, we propose a notion of signature for reduction monads,
which we call ``reduction signatures''.
Each reduction signature comes equipped with the category of its models: such a model is a reduction monad ``acted upon'' by the signature. 
A reduction signature may be understood as a virtual presentation:
when an initial model exists, it inherits a kind of presentation given by the action of the signature,
and we say that the signature is \emph{effective}. 
Our main result (Theorem~\ref{thm:alg-2-sig-initial}) provides a natural criterion for a reduction signature to be effective.
Our main examples are variants of the lambda calculus (Section~\ref{ss:examples-reduction-monads} and Example~\ref{ex:variants-lc}).

In summary, our notions of reduction monad and (effective) reduction signature
\begin{itemize}
 \item allow  ``dynamic'' reduction systems (as opposed to ``static'' ones, where reduction is understood as an equivalence relation); 
 \item cover (some) higher-order languages;
 \item allow reduction systems that are not fully congruent, e.g., weak head reduction in lambda calculus; and
 \item allow reduction systems  that are proof-relevant.
 \end{itemize}

While the first three features listed above are definitely positive, we feel the need to discuss the advantages of the fourth one. 
Basically, one can model reductions either via preorders or relations,
or via multigraphs, and we choose multigraphs.
On the one hand, multigraphs are appropriate to avoid problems with the well-know phenomena of \emph{syntactic accidents} (see~\cite[Example~2.2.9 and Section~8.2]{terese} for a classic reference).
On the other hand, the
approaches via preorders or relations can be easily recovered as special cases.
Indeed, note that the category of preordered sets (as considered, e.g., in \cite{ahrens_relmonads}) is a reflective subcategory of the category of relations  (as considered, e.g., in \cite{DBLP:journals/jlp/Plotkin04a}), which is again a reflective subcategory of the category of multigraphs.
Then, the corresponding adjunctions make it easy to customize our formalism to deal with preordered sets or relations instead of multigraphs.

\subsection{Related Work}

The search for a mathematical notion of programming language goes back at least to Turi and Plotkin \cite{DBLP:conf/lics/TuriP97} who coined the name ``Mathematical Operational Semantics'' and explained how known classes of well-behaved rules for structural operational semantics (SOS) \cite{DBLP:journals/jlp/Plotkin04a}, such  as GSOS \cite{DBLP:conf/popl/BloomIM88}, can be categorically understood via distributive laws and bialgebras (see also \cite{DBLP:journals/tcs/Klin11}). 
Their initial framework did not cover variable binding, and several authors have proposed variants which do \cite{DBLP:conf/lics/FioreT01,DBLP:conf/lics/FioreS06,DBLP:conf/lics/Staton08}, treating examples like the \(\pi\)-calculus.  More recently T.\ Hirschowitz \cite{DBLP:journals/pacmpl/Hirschowitz19} proposed an alternative categorical approach to SOS allowing variable binding.
However, none of these approaches covers higher-order languages like the lambda calculus.

Meanwhile, similar research has been done from the point of view of ``Rewriting Systems'' or ``Equational Systems''.
(A rewriting system is intended to specify a  fully congruent relation, while an equational system is intended to specify a fully congruent equivalence relation.)
In particular:
\begin{itemize}
\item Lüth and Ghani \cite{DBLP:conf/ctcs/LuethG97} interpreted a class of rewriting systems  (without considering bindings) in monads in the category of preordered sets.
\item  Hamana \cite{Hamana:2003:TRV:888251.888266} studied his ``binding term rewriting system (BTRS)'' via preorder-valued functors.
\item 
Building upon \cite{FPT}, Fiore and Hur \cite{FH07} have studied a large class of ``term equational systems'', covering in particular the lambda calculus (see also \cite{DBLP:conf/mfcs/FioreM10}).
In our recent work \cite{FSCD2019}, we have proposed a variant of their work in terms of signatures for monads; this is reviewed below in Section~\ref{sss:1-sigs} since it is the starting building block for the present work.
\item T.~Hirschowitz \cite{DBLP:journals/corr/Hirschowitz14} approached higher-order rewriting systems via Cartesian closed 2-categories.
\item Ahrens \cite{ahrens_relmonads} provided a signature for  the lambda calculus viewed as a monad relative to the inclusion of sets into preorders mapping a set to its discrete preorder.
\end{itemize}
This last work  was a starting point of the present one and we stress two main differences.
The first difference is technical but important: our signatures are built from
modules with values in $\Cat{Set}$ (not in the category of multigraphs) while in~\cite{ahrens_relmonads} signatures involve modules with values in $\Cat{Preorder}$ (not in $\Cat{Set}$). This latter choice is responsible for the fact that full congruence of reductions is hard-coded in~\cite{ahrens_relmonads}.
A second difference is that the format for reduction rules considered in \cite{ahrens_relmonads} does not allow rules with hypotheses, as present for instance in the ``non-full'' congruence rules which specify the head-\(\beta\)-reduction (see the table at the end of Section~\ref{ss:ex-reduction-rules}). 

\subsection{Plan of the Paper}
In Section~\ref{sec:review}, we review those notions from our previous work \cite{ahrens_et_al:LIPIcs:2018:9671,FSCD2019} that we build on in the present work.
In Section~\ref{sec:reduction-monads}, we define the category of reduction monads.
In Section~\ref{sec:reduction-rules}, we give our notion of reduction rules. Building on this notion, we define signatures for reduction monads---\emph{reduction signatures}---in Section~\ref{sec:signatures}.
There, we also state our main result (Theorem~\ref{thm:alg-2-sig-initial}), which provides a simple criterion for such a signature to be effective.
Section~\ref{s:effective-alg-2sig} is devoted to the proof of this result.
Then, in Section~\ref{s:ex-lex}, we give a detailed example of a reduction signature specifying Kesner's lambda calculus with explicit substitutions \cite{lambda-ex}. 
Finally, in Section~\ref{sec:recursion},
we explain the recursion principle which, as usual, can be derived from initiality in our categories of models. We then use this recursion principle to specify a translation from Kesner's lambda calculus with explicit substitutions to the pure lambda calculus.

\section{\texorpdfstring{$\Cat{C}$}{C}-modules and signatures for monads}
\label{sec:review}

The present work is devoted to reduction monads and their signatures.
The first building block of a reduction monad is a monad, and accordingly, the first building block of one of our signatures for reduction monads will be a signature for monads. 
In the present section, we review monads  and their signatures as introduced in \cite{FSCD2019}. 
In that paper, one can also find a review of related work on the generation of monads.

A signature prescribes a monad by specifying:
\begin{enumerate}
\item a family of \emph{constructions};
\item a family of \emph{equations} among these constructions.
\end{enumerate}
As an example, in the case of the lambda calculus, the constructions are application and abstraction (in the following denoted $\app$ and $\abs$).
Equations that can be considered over those constructions are the $\beta$- and $\eta$-equations.
Notice, however, that later in this paper we will opt for considering $\beta$ and $\eta$ as \textbf{reductions} (as opposed to equivalences, or equations).

In Section~\ref{ss:monads-modules} we review the notions of monads and modules (on sets).
In Section~\ref{ss:c-modules} we introduce the notion of $\CC$-module for a category $\CC$.
This notion is used in our definition of signature.
In our formalism, constructions are specified by \emph{1-signatures} (reviewed in Section~\ref{sss:1-sigs}), and equalities between constructions specified by 1-signatures are specified by \emph{equations} (reviewed in Section~\ref{sss:2-sigs}). 
A 2-signature is a pair consisting of a 1-signature and a family of equations over that signature.
We refer to both 1- and 2-signatures simply as \textbf{signatures for monads} when the distinction between the two notions is not relevant.

\subsection{Monads and Modules}
\label{ss:monads-modules}

In this section we review the notions of monad and module to the extent required in this work.
We write $g\circ f$ for the composition of morphisms $f : A \to B$ and $g : B \to C$ in any category. Similarly, functor composition is written $G \cdot F$.
For the purpose of this work, we restrict ourselves to the category $\Mon$ of monads over the category $\Set$ of sets.  

A \textbf{monad} consists of a function $R:\Set \to \Set$, a family of functions $\et_X:X\to R(X)$, and a family of functions $\_\monsubst{\_} : RX \times (X \to RY) \to RY$ called \emph{substitution},
satisfying the equations $t\monsubst{\et_X} = t$, $\et_X(x)\monsubst{f} = f(x)$,
$t\monsubst{f}\monsubst{g} = t\monsubst{x \mapsto  f(x)\monsubst{g}}$ for $t \in RX$, $f : X\to RY$ and $g : Y \to RZ$.
Sometimes, we identify an element $x\in X$ with the corresponding term $\et_X(x) \in R(X)$ when no confusion can arise.
A \textbf{module $M$ over a monad $R$} is a function $\Set \to \Set$ with a family of functions
$\_\monsubst{\_} : MX \times (X \to RY) \to MY$ called \emph{substitution}, satisfying the equations
$t\monsubst{\et_X} = t$ and $t\monsubst{f}\monsubst{g} = t\monsubst{x\mapsto f(x)\monsubst{g}}$ for $t \in MX$, $f : X\to RY$ and $g : Y \to RZ$. 
Note that we use the same notation for the substitution of monads and of modules;
in the last equation, the inner substitution on the right hand side is the one from the monad $R$.
A monad $R$ (resp.\ a module $M$ over a monad $R$) gives rise to a functorial action $RX \times (X \to Y) \to RY$ (resp.\ $MX \times (X\to Y) \to MY$) by $(t,g) \mapsto t\monsubst{\et_Y \circ g}$.
A \textbf{morphism $R \to S$ of monads} is a family of functions $\alpha_X : RX \to SX$ 
commuting with $\et$ and $\_\monsubst{\_}$ of $R$ and $S$, in the sense that
$\alpha(\eta^R(x)) = \eta^S(x)$ and $\alpha(t\monsubst{f}) = \alpha(t)\monsubst{\alpha \circ f}$ for $x\in X$, $t\in RX$ and $f : X \to RY$. 
With the obvious composition operation, this yields the \textbf{category $\Mon$ of monads on sets}.
A \textbf{morphism of modules $M \to N$ over a monad $R$} is a family of functions $\beta_X : MX \to NX$ commuting with $\_\monsubst{\_}$ of $M$ and $N$ in the sense that $\beta(t\monsubst{f}) = \beta(t)\monsubst{f}$ for $t\in M(X)$ and $f : X \to RY$.
Monad and module morphisms are (in particular) natural transformations.

If $M$ is a module over $S$ and $\alpha:R\to S$ is a monad morphism, then we denote $\alpha^* M$ the $R$-module with the underlying function $X \mapsto M(X)$ on sets, and with substitution induced by that of $M$ by $(t,f) \mapsto t\monsubst{\alpha_Y \circ f}$; this is called the \textbf{reindexing of $M$ along $\alpha$}.

We denote by $\Mod(R)$ the \textbf{category of modules over a fixed monad $R$}, and by $\LMod$ the \textbf{total category of modules}, fibered over the category $\Mon$ of monads
\cite[Proposition 7]{ahrens_et_al:LIPIcs:2018:9671}: objects of $\LMod$ are pairs of a monad and a module over it, and a morphism from $(R,M)$ to $(S,N)$ is a monad morphism $\alpha:R\to S$ together with a $R$-module morphism between $M$ and $\alpha^* N$.

A prominent example of monad, used heavily in the rest of this paper, is the (untyped, syntactic) lambda calculus: given a set $X$ (to be understood as a collection of free variables), $\LC(X)$ is the set of lambda terms with free variables in $X$ modulo $\alpha$-equivalence, and the usual parallel, capture-avoiding substitution endows $\LC$ with a monadic structure (see \cite{Alt-Reus}).

Modules and morphisms of modules are used to specify constructions and equations within a signature.  To this end, we recall some basic module constructions that are used compositionally.

First of all, a simple observation is that every monad $R$ is (trivially) a module over itself.  We denote $\Theta(R)$ (sometimes just $R$) the monad $R$ regarded as an $R$-module.  The assignment $R\mapsto \Theta(R)$ is the object map of the \textbf{tautological} section $\Theta: \Mon \to \LMod$.

Next, in the category $\Mod(R)$ of modules over a fixed monad $R$, we have arbitrary limits and
colimits, which are computed pointwise.
In particular, we have the initial and final modules and, for two $R$-modules $M$ and $N$, the binary product $M\times N$ and the binary coproduct $M + N$.

Given an $R$-module $M$, we have the \textbf{derived} $R$-module $M'$ defined as follows.
We denote by $X + \{*\}$ the disjoint union of a set $X$ with a one-element set.  Then, we set $M'(X) := M(X + \{*\})$.
The $R$-module structure on $M$ induces naturally an $R$-module structure on $M'$.
This derivation yields an endofunctor on the category of modules $\Mod(R)$ for any monad $R$.
It can be iterated; we call $M^{(k)}$ the $k$-th derivative of $M$.

\subsection{\texorpdfstring{$\CC$}{C}-Modules}
\label{ss:c-modules}

In this section, $\CC$ is a category equipped with a functor $U : \CC\to\Mon$.
We introduce the notion of $\CC$-module, generalizing $\Sigma$-modules of~\cite{FSCD2019}.

\begin{definition}
A \textbf{$\CC$-module} is 
  a functor $T$ from the category $\CC$ 
   to the category $\LMod$ commuting with
  the functors to the category $\Mon$ of monads,
 \[
  \begin{xy}
   \xymatrix{
                    \CC \ar[rd]_U \ar[rr]^{T}  &        &   \LMod \ar[dl] \
                    \\
                                     & \Mon
   }
  \end{xy}
 \]
 More concretely, a $\CC$-module $T$ maps an object $c$ of $\CC$ to a module
 $T(c)$ over $U(c)$ and a morphism $f:c\to c'$ to a $U(c)$-module morphism
 $T(f):T(c)\to U(f)^* T(c')$.
 In this paper, $U$ will arise as a forgetful functor to the category $\Mon$ of monads and will be tacitly omitted.
\end{definition}

\begin{definition}
The \textbf{tautological 
$\CC$-module $\Theta$} maps an object $c$ of $\CC$ to the module $\Theta(U(c))$.
\end{definition}
\begin{definition}
Let $A$ be a $\CC$-module.
The \textbf{derivative of $A$} is the $\CC$-module $A'$ mapping an object $c$ of
$\CC$ to the $U(c)$-module $A(c)'$.
\end{definition}
  
We typically apply (iteratively) the derivation construction to
the tautological $\CC$-module $\Theta$:
The $\CC$-module $\Theta'$ maps an object $c$ of $\CC$ to the $U(c)$-module $U(c)'$,
and the $\CC$-module $\Theta^{(n)}$ maps $c$ to the $U(c)$-module $U(c)^{(n)}$.
  
We introduce the \emph{category of $\CC$-modules}.
Importantly, as we will see later, our \emph{term-pairs} are built from 
pairs of parallel morphisms of $\CC$-modules for a suitable category $\CC$.
  
\begin{definition}
Let $S$ and $T$ be $\CC$-modules.
  A \textbf{morphism of $\CC$-modules} from $S$ to $T$ is a natural
  transformation from $S$ to $T$ which becomes the identity of $U$ when postcomposed with the
  forgetful functor from the category of modules $\LMod$ to the category of monads. 
\end{definition}  
\begin{proposition}
As defined above, $\CC$-modules and their morphisms, with the obvious composition and identity,
 form a category.
\end{proposition}

\begin{definition}
Let $A$ and $B$ be $\CC$-modules, and let $f:A\to B$ be a morphism of
$\CC$-modules. The \textbf{derivative of $f$} is the $\CC$-module morphism from
$A'$ to $B'$ mapping an object $c$ of $\CC$ to the $U(c)$-module morphism $f_c'$.
\end{definition}

\subsection{1-Signatures}
\label{sss:1-sigs}

In this section we review the 1-signatures introduced in \cite{ahrens_et_al:LIPIcs:2018:9671}.
\begin{definition}[1-Signatures and their models, {\cite[Defs.\ 14 and 27]{ahrens_et_al:LIPIcs:2018:9671}}]
 A \textbf{1-signature} is a $\Mon$-module, hence 
 it is just a section to the forgetful functor $\pi : \LMod \to \Mon$.

 Let $\Sigma : \Mon \to \LMod$ be a 1-signature. A \textbf{model of $\Sigma$}
 is a pair of a monad $R$ and a module morphism $\Sigma(R) \to R$, called an \textbf{action of $\Sigma$ in $R$}.
   A \textbf{morphism  from $(R, r)$ to $(S, s)$} is a
  morphism of monads $m : R \to S$ making the following diagram of
  $R$-modules commutes:
  \[
    \xymatrix{
      **[l] \Sigma(R) \ar[r]^{r}\ar[d]_{\Sigma(m)} & **[r] R \ar[d]^m \\
      **[l] m^* (\Sigma(S)) \ar[r]_{m^*s} & **[r] m^* S}
  \]
  Recall that the module $m^*S$ is the reindexing of $S$ along $m$ in the fibration $\LMod \to \Mon$.
  We call $\Mon^{\Sigma}$ the induced category of models of $\Sigma$.
\end{definition}

The desired property for a signature is \enquote{effectivity}:

\begin{definition}
 Given a 1-signature $\Sigma$,
 the initial object in $\Mon^\Sigma$, if it exists, is denoted by
 $\hat{\Sigma}$.
  In this case, 
   the 1-signature $\Sigma$ is called
   \textbf{effective}.
\end{definition}

\begin{example}
 \label{ex:1-sig-lc}
 Consider the 1-signature $\Sigma_\LC$ given on objects by $R \mapsto R \times R + R'$. 
A model of this signature is a triple $(R, \app, \abs)$ with $\app : R \times R \to R$ and $\abs : R' \to R$. The initial model is the triple $(\LC,\app,\abs)$ of the monad $\LC$ of lambda terms, with application and abstraction as module morphisms.
Note that $\alpha$-equivalence corresponds to syntactic equality of lambda terms.
\end{example}
\begin{example}[Non-effective 1-signature]
 \label{ex:non-effective-one-sig}
 The 1-signature mapping a monad $R$  to the $R$-module $\mathcal{P}\cdot R$,
 where $\mathcal{P}$ is the powerset functor, is not effective, for cardinality reasons.
\end{example}
\begin{rem}
\label{rk:sigma-arity-1-sig}
Let $\CC$ be a category equipped with a functor $U : \CC\to\Mon$.
Any 1-signature $\Psi$ (i.e. any functor from $\Mon$ to $\LMod$ commuting with the forgetful functors to $\Mon$) induces a $\CC$-module still denoted $\Psi$, by precomposition with $U$.
\end{rem}

\subsection{Equations and 2-Signatures}
\label{sss:2-sigs}

In this section we review the 2-signatures introduced in \cite{FSCD2019}.

Let $\Sigma$ be a 1-signature.
In the following, we use the term ``$\Sigma$-module'' for $\Mon^\Sigma$-modules.

\begin{definition}[$\Sigma$-equation, {\cite[Def.\ 8]{FSCD2019}}]
  We define a \textbf{$\Sigma$-equation} 
  to be a pair of
  parallel morphisms of $\Sigma$-modules. 
\end{definition}

\begin{definition}[2-signature, {\cite[Def.\ 12]{FSCD2019}}]
  A \textbf{2-signature} is a pair $(\Sigma,E)$ of a 1-signature $\Sigma$ and a
  family $E$ of $\Sigma$-equations. 
\end{definition}

\begin{definition}[Model of a 2-signature, {\cite[Def.\ 17]{FSCD2019}}]
  We say that a model $M$ of a 1-signature $\Sigma$ 
  \textbf{satisfies the $\Sigma$-equation $(e_1, e_2)$} if $e_{1}(M) = e_{2}(M)$.
  If $E$ is a family of $\Sigma$-equations, we say that a
  model $M$ of $\Sigma$ \textbf{satisfies $E$} if $M$ satisfies each
  $\Sigma$-equation in $E$.

  Given a monad $R$ and a 2-signature $\Upsilon=(\Sigma,E)$, an \textbf{action
  of $\Upsilon$ in $R$} is an action of $\Sigma$ in $R$ such that the induced
model satisfies all the equations in $E$.

  For a 2-signature
  $(\Sigma,E)$, we define the \textbf{category $\Mon^{(\Sigma,E)}$ of
  models of $(\Sigma, E)$} to be the full subcategory of the
  category of models of $\Sigma$ whose objects are
  models of $\Sigma$ satisfying $E$, or equivalently, monads 
  equipped with an action of $(\Sigma , E)$.
\end{definition}

As for 1-signatures, we define a notion of \emph{effectivity}:
\begin{definition}
  A 2-signature $\Upsilon$ is said to be
  \textbf{effective}
  if its category of models $\Mon^{\Upsilon}$ has an initial object,
  denoted $\hat{\Upsilon}$.
\end{definition}

As an example, consider the signature $\Sigma_\Lambda$ obtained from $\Sigma_\LC$ by imposing the $\beta$- and $\eta$-equalities.
Then, the monad $\Lambda$ of lambda calculus modulo $\beta\eta$-equality is the initial object of the category of models $\Mon^{\Sigma_\Lambda}$ \cite[Theorem 3]{Hirschowitz-Maggesi-2010}.

As 1-signatures are particular 2-signatures with an empty set of equations, Example~\ref{ex:non-effective-one-sig} yields a non-effective 2-signature.
Another example is a 2-signature containing the equation
\( \mathsf{inl},\mathsf{inr} : \Theta \rightrightarrows \Theta + \Theta \)
given by the left and right inclusions.
An important class of effective signatures for monads is the one of \emph{algebraic}
(2-)signatures~\cite[Theorem 32]{FSCD2019}, which covers both $\Sigma_\LC$ and $\Sigma_\Lambda$ as particular cases.

\section{Reduction monads}
\label{sec:reduction-monads}

Here below, we define \emph{the category of reduction monads} in Section~\ref{ss:def-reduction-monad}.
We also consider some examples of reduction monads, in Section~\ref{ss:examples-reduction-monads}.

\subsection{The Category of Reduction Monads}
\label{ss:def-reduction-monad}

\begin{definition}
A \textbf{reduction monad} $R$ is given by:
\begin{enumerate} 
    \item a monad on sets, that we still denote by $R$, or by $\monadfromred{R}$ when we want to be explicit;
    \item an $R$-module $\Redof R$ (the module of \emph{reductions});
    \item a morphism of $R$-modules $\redof R : \Redof R \to R\times R$ {(\emph{source} and \emph{target} of rules)}.
\end{enumerate}

We set $\source{R} := \pi_1 \circ \redof{R} : \Redof{R} \to R$, and
$\target{R} := \pi_2 \circ \redof{R} : \Redof{R} \to R$

\end{definition}

For a reduction monad $R$, a set $X$, and elements $s,t\in R(X)$,
we think of the fiber $\redof{R}(X)^{-1}(s,t)$ as the set of ``reductions from $s$ to $t$''.
We sometimes write $m:\redfib{s}{t}:R(X)$, or even $m : \redfib{s}{t}$ when there is
no ambiguity, instead of $m \in \redof{R}(X)^{-1}(s,t)$.

\begin{rem}
 \label{r:red-mon-proof-relevant}
 Note that for a given reduction monad $R$, set $X$, and $s,t : R(X)$, there can be multiple 
 reductions from $s$ to $t$, that is, the fiber $\redfib{s}{t}$ is not necessarily a subsingleton.
\end{rem}

\begin{rem}
\label{r:congr-weak-subst}
Let $R$ be a reduction monad, $X$ and $Y$ two sets, $f: X \to R(Y)$ a map, and $u$ and $v$ two elements of $R(X)$ related by
$m : \redfib{u}{v}$.
The module structure on $\Redof{R}$ yields a reduction denoted $m\monsubst{f}$ between $u\monsubst{f}$ and $v\monsubst{f}$.

However, if we are given two maps $f$ and $g$, and for all $x\in X$, 
a reduction $m_x:\redfib{f(x)}{g(x)}$, then it does not follow that there is a reduction between $u\monsubst{f}$ and $u\monsubst{g}$. 
This leaves the door open for non-congruent reductions.
\end{rem}

Our main examples of reduction monads are given by variants of the lambda calculus.
We have collected these examples in Section~\ref{ss:examples-reduction-monads}.

\begin{definition}
 A \textbf{morphism of reduction monads} from $R$ to $S$ is 
  given by a pair $(f, \alpha)$ of
  \begin{enumerate}
      \item a monad morphism $f:R\to S$, and
      \item a natural transformation $\alpha:\Redof{R} \to \Redof{S}$ 
     \setcounter{nameOfYourChoice}{\value{enumi}}
\end{enumerate}
satisfying the following two conditions:
 \begin{enumerate}
\setcounter{enumi}{\value{nameOfYourChoice}}
    \item $\alpha$ is an $R$-module morphism between $\Redof{R}$ and the
      reindexed module 
      $f^* \Redof{S}$ of $\Redof{S}$ along $f$, and
           \label{enum:reduction-mor-module-mor}
    \item the square
            \[
            \xymatrix{
                   \Redof{R} \ar[r]^\alpha \ar[d]_{\redof{R}}  &  \Redof{S} \ar[d]^{\redof{S}} \\
                  R \times R \ar[r]_{f\times f} & S \times S
                    }
               \]
           commutes in the category of functors and natural transformations.
           \label{enum:reduction-mor-commutes}
 \end{enumerate}
\end{definition}
In Section~\ref{sec:recursion} we specify morphisms of reduction monads via 
a recursion principle.

Intuitively, a morphism $(f,\alpha)$ as above maps terms of $R$ to terms of $S$ via $f$,
and reductions of $R$ to reductions of $S$ via $\alpha$.
Condition~\ref{enum:reduction-mor-module-mor} states compatibility of the map of reductions with substitution:
$\alpha(m\monsubst{g}) = \alpha(m)\monsubst{f_Y\circ g}$ for any
reduction $m : \redfib{u}{v}$ and any map $g:X\to R(Y)$.
Condition~\ref{enum:reduction-mor-commutes} states preservation of source and target by the 
map of reductions: a reduction $m : \redfib{u}{v}$ between elements of $R(X)$ is mapped by $\alpha$
to a reduction $\alpha(m) : \redfib{f_X(u)}{f_X(v)}$.

\begin{proposition}[Category of reduction monads]
  Reduction monads and their morphisms, with the obvious composition and identity,
  form a category $\GMon$, equipped with a forgetful functor to the category of monads.
\end{proposition}

It turns out that reduction monads are the same as monads relative to the free functor from sets to multigraphs
(for the definition of relative monads, see
\cite[Definition~2.1]{DBLP:journals/corr/AltenkirchCU14}):
\begin{theorem}
  \label{t:structure-thm}
  The category of reduction monads is isomorphic to the category of monads relative to the functor mapping a set to its discrete multigraph.
\end{theorem}
\begin{proof}
This is obvious after unfolding the definitions.
\end{proof}

\subsection{Examples of Reduction Monads}
\label{ss:examples-reduction-monads}

We are interested in reduction monads with underlying monad $\LC$,
the monad of syntactic lambda terms specified in Example~\ref{ex:1-sig-lc}.
We start with a detailed simple example. The other ones
(Examples~\ref{ex:lambda-whb},~\ref{ex:lc-congruent-beta},
and~\ref{ex:lc-parbeta}) are more informal. Reduction
signatures specifying them will be given later in Example~\ref{ex:variants-lc}.

\begin{example}[Lambda calculus with \textbf{top-$\beta$-reduction}]
\label{ex:lambda-calculus-top-beta-reduction}
  Consider the reduction monad $\LC_{\text{top-}\beta}$ given as follows:
  \begin{enumerate}
      \item the underlying monad is $\LC$;
      \item $\Redof{R}$ is the module $\LC' \times \LC$;
      \item $\redof{R}(X)$ is the morphism $(u,v) \mapsto \bigl(\app(\abs(u),v), u\monsubst{*:=v}\bigr)$.
  \end{enumerate}
\end{example}
\begin{example}[Lambda calculus with \textbf{weak head $\beta$-reduction}]
  \label{ex:lambda-whb}
 We introduce the reduction monad $\LCwhb$.
 A reduction $m\in\Redof{\LCwhb}(X)$ in the reduction monad $\LCwhb$ is a
 leftmost $\beta$-reduction in a chain of applications. Using the standard
 syntax of lambda calculus, $(\lambda
 x.t)\ u_1\ u_2\ \dots\ u_n$ reduces to $t\monsubst{x := u_1}\ u_2\ \dots\ u_n$.
\end{example}
\begin{example}[Lambda calculus with \textbf{congruent $\beta$-reduction}]
  \label{ex:lc-congruent-beta}
 We introduce the reduction monad $\LCb$.
 A reduction $m\in\Redof{\LCb}(X)$ in the reduction monad $\LCb$ is ``one step'' of $\beta$-reduction, anywhere in the source term.
\end{example}
\begin{example}[Lambda calculus with \textbf{parallel $\beta$-reduction}]
  \label{ex:lc-parbeta}
 A reduction $m\in\Redof{\LCbpar}(X)$ in the reduction monad $\LCbpar$
 is the simultaneous $\beta$-reduction of a ``parallel'' set of redexes
in the source term.
\enquote{Parallel} here means that the subtrees of the contracted redexes are disjoint.
\end{example}

Next, we consider the \emph{closure under identity and composition of reductions}.
\begin{definition}
  \label{d:transitive-closure-redmonad}
Given a reduction monad $R$, we define the reduction monad $R^*$ as follows:
\begin{enumerate}
    \item the underlying monad on sets is still $R$;
    \item the $R$-module $\Redof{R^*}$ is defined as follows.
	  For $n \in \N$ we define the module $\Redof{R}^n$ of ``$n$ composable reductions'',
	  namely as the limit of the diagram
	  \[
	   \begin{xy}
	    \xymatrix@C=3em@R=1em{
	          & \Redof{R} \ar[ld]_{\source{R}} \ar[rd]^{\target{R}} &     & \Redof{R} \ar[ld]_{\source{R}} \ar[rd]^{\target{R}} & & \ldots \ar[ld]_{\source{R}}
	          \\
	       R  &           & R   &           & R
	    }
	   \end{xy}
	  \]
          with $n$ copies of $\Redof{R}$ (and hence $n+1$ copies of $R$).
          We obtain $n+1$ projections $\pi_i : \Redof{R}^n \to R$,
          and we call $p_n := (\pi_0,\pi_n) :  \Redof{R}^n \to R \times R$.
	  We set $\Redof{R^*} := \coprod_{n} \Redof{R}^n$.
    \item the module morphism is $\redof{R^*} := [p_n]_{n \in \N} : \coprod_{n} \Redof{R}^n \to R \times R$
           the universal morphism induced by the family $(p_n)_{n\in N}$.
\end{enumerate}
\end{definition}

\begin{example}[The reduction monad of the lambda calculus]
\label{ex:real-lambda-calculus}
  The \textbf{reduction monad of the lambda calculus} is defined to be the reduction monad
  $\LCclot$.
\end{example}

In Section~\ref{sec:signatures} we introduce signatures that allow for the specification of reduction monads. 

\section{Reduction rules}
\label{sec:reduction-rules}
In this section, we define an abstract notion of
reduction rule over a signature for monads $\Sigma$ (Section~\ref{ss:reduction-rules}).
We first focus, in Section~\ref{ss:ex-congr-rule-app}, on the example
of the congruence rule for the application construction in the signature $\Sigma_\LC$ (cf.\ Example~\ref{ex:1-sig-lc}) for the monad of the lambda calculus, in order to motivate the definitions.
The purpose of a reduction rule over $\Sigma$ is to be \enquote{modeled} in
a reduction monad equipped with an action of $\Sigma$ (this is what we will call a \emph{reduction $\Sigma$-model} in Section~\ref{ss:reduction-sigma-monads}).
We make this notion of model precise in Section~\ref{ss:action-reduction-rule}, as an action of the reduction rule in the reduction $\Sigma$-model.
Finally, we give a protocol for specifying reduction rules in Section~\ref{ss:protocole-reduction-rules} that we
apply in Section~\ref{ss:ex-reduction-rules} to some examples.

\begin{notation}
In the following, we use the shorter terminology of $\Sigma$-modules for
$\Mon^\Sigma$-modules, when $\Sigma$ is any signature for monads (either a
1-signature or a 2-signature).
\end{notation}

\subsection{Example: Congruence Rule for Application}
\label{ss:ex-congr-rule-app}

We give some intuitions of the definition of reduction rule with the example of the congruence rule for application, given, e.g., in Selinger's lecture notes \cite{selinger-lc}, as follows:
\begin{equation*}
\label{pt:congr-app-selinger}
\frac{
\redrule{T}{T'} 
\qquad
\redrule{U}{U'}
}
{\redrule{\app(T,U)}{\app(T',U')} }
\end{equation*}

This rule is parameterized by four \emph{metavariables}: $T$, $T'$, $U$, and $U'$. The conclusion and the hypotheses are given by pairs of terms built out of these metavariables. 

We formalize this rule as follows: for any monad $R$ equipped with an application operation $\app : R\times R\to R$, we associate a module of metavariables $\mathcal{V}(R)=R\times R \times R \times R$, one factor for each of the metavariables $T$, $T'$, $U$, and $U'$. 
Each hypothesis or conclusion is described by a parallel pair of morphisms from $\mathcal{V}(R)$ to $R$: for example, the conclusion $c_R: \mathcal{V}(R) \to R$ maps a set $X$ and a quadruple $(T,T',U,U')$ to the pair $(\app(T,U),\app(T',U'))$. These assignments are actually functorial in $R$, and abstracting over $R$ yields our notion of \emph{term-pair} over the $\Sigma$-module $\mathcal{V}$, as morphisms from $\mathcal{V}$ to $\Theta\times\Theta$, where $\Sigma$ is any signature including a single first-order binary operation $\app$ (for example, $\Sigma_\LC$).
The three term-pairs, one for each hypothesis and one for the
conclusion, define the desired reduction rule.

Now, we explain in which sense such a rule can be modeled in a reduction monad $R$: intuitively, it means that for any set $X$, any quadruple $(T,T',U,U')\in \mathcal{V}(R)$, any reductions $s:\redfib{T}{T'}$ and $t:\redfib{U}{U'}$, there is a reduction $\appcongr(s,t):\redfib{\app(T,U)}{\app(T',U')}$.
Of course, this only makes sense if the monad $\monadfromred{R}$ underlying the reduction monad is equipped with an application operation, that is, with an operation of $\Sigma_\app$. We will call such a structure a \emph{reduction $\Sigma_\app$-model} (see Section~\ref{ss:reduction-sigma-monads}).

\subsection{Definition of Reduction Rules}
\label{ss:reduction-rules}

In this subsection, $\Sigma$ is a signature for monads. We present our notion of \emph{reduction rule over $\Sigma$},
from which we build \emph{reduction signatures} in Section~\ref{sec:initiality}.

We begin with the definition of \emph{term-pair}, 
alluded to already in Section~\ref{ss:ex-congr-rule-app}:

\begin{definition}
Given a $\Sigma$-module $\mathcal{V}$,
  a \textbf{term-pair from $\mathcal{V}$}
  is a pair $(n,p)$ of a natural number $n$ and 
  a morphism of $\Sigma$-modules $p:\mathcal{V} \to \Theta^{(n)}\times \Theta^{(n)}$. 
\end{definition}

Many term-pairs are of a particularly simple form, namely a pair of projections,
which intuitively picks two among the available metavariables.
Because of their ubiquity, we introduce the following notation:
\begin{definition}
  \label{d:pair-projection}
  Let $n_1$, \dots, $n_p$ be a list of natural numbers.
  For $i,j\in\{1 ,\dots p\}$,
  we define the projections $\pi_i$ and $\pi_{i,j}$ as the following $\Sigma$-module
  morphisms, for any signature $\Sigma$:
  \[
  \begin{array}{rl}
     \pi_{i,j} :& \Theta^{(n_1)}\times\dots\Theta^{(n_p)}\to\Theta^{(n_i)}\times\Theta^{(n_j)} \\
     \pi_{i} :& \Theta^{(n_1)}\times\dots\Theta^{(n_p)}\to\Theta^{(n_i)} 
  \end{array}
  \qquad
  \begin{array}{rc}
      \pi_{i,j,R,X}  (T_1,\dots,T_p)  =& (T_i, T_j) \\
        \pi_{i,R,X}  (T_1,\dots,T_p)  =& T_i
  \end{array}
  \]
  \end{definition}

Some term-pairs, such as the conclusions of the congruence rules for application and abstraction, are more complicated: intuitively, they are obtained by applying term constructions to metavariables.

\begin{example}[term-pair of the conclusion of the congruence for application]
\label{ex:red-judg-concl-app}
  The term-pair corresponding to the conclusion $\redrule{\app(T,U)}{\app(T',U')}$ of congruence for application (Section~\ref{ss:ex-congr-rule-app})
  is given by $(0,\pcon)$, on the $\Sigma_\LC$-module $\Theta^4$. 
  Here, we have
  \begin{align*}
      &\pcon : \mathcal{V} \to \Theta \times \Theta
      \\
      &\pcon_{R,X}(T,T',U,U') := \Bigl(\app_{R,X}(T,U), \app_{R,X}(T',U') \Bigr)
  \end{align*}
  More schematically:
   \[
      \xymatrix@C=25pt@R=2pt{
  \pcon : \Theta^4 
  \ar[rrrr]^-{\app\circ\pi_{1,3},\app\circ\pi_{2,4}}
  &&&
  &
  \Theta\times\Theta
    }
     \]
\end{example}

We now give our definition of \emph{reduction rule},
making precise the intuition developed in Section~\ref{ss:ex-congr-rule-app}.

\begin{definition}
\label{def:elementary-arity}
 A \textbf{reduction rule $\mathcal{A}=(\mathcal{V},(n_i,\phyp_i)_{i\in I},(n,\pcon))$ over $\Sigma$} is given by:
 \begin{itemize}
     \item Metavariables:
     a $\Sigma$-module $\mathcal{V}$ of metavariables, that we sometimes denote by $\MVar_\mathcal{A}$;
     \item Hypotheses: a finite family of term-pairs $(n_i,\phyp_i)_{i\in I}$ from $\mathcal{V}$;
     \item Conclusion: a term-pair $(n,\pcon)$ from $\mathcal{V}$.
 \end{itemize}
  \end{definition}
  \begin{example}[Reduction rule for congruence of application]
  \label{ex:congr-app-rule}
  The reduction rule $\mathcal{A}_\appcongr$ for congruence of application (Section~\ref{ss:ex-congr-rule-app}) is defined as follows:
   \begin{itemize}
     \item \MMM: $\mathcal{V}=\Theta^4$ for the four metavariables $T$, $T'$, $U$, and $U'$;
     \item \HHH: 
        Given by two term-pairs $(0,\phyp_1)$ and $(0,\phyp_2)$:
         \begin{equation*}
  \xymatrix@C=25pt@R=2pt{
  \phyp_1 : \Theta^4 
  \ar[rr]^-{\pi_{1,2}}
  &
  &
  \Theta\times\Theta
    }
    \qquad
    \xymatrix@C=25pt@R=2pt{
  \phyp_2 : \Theta^4 
  \ar[rr]^-{\pi_{3,4}}
  &
  &
  \Theta\times\Theta
    }
\end{equation*}
     
     \item \CCC: Given by the term-pair $(0,\pcon)$ of Example~\ref{ex:red-judg-concl-app}.
 \end{itemize}
  \end{example}

More examples of reduction rules are given in Section~\ref{ss:ex-reduction-rules}.

\subsection{Reduction \texorpdfstring{$\Sigma$}{Sigma}-Models}
\label{ss:reduction-sigma-monads}

As already said, the purpose of a reduction rule is to be modeled in a reduction monad $R$.
However, as the hypotheses or the conclusion of the reduction rule may refer to some operations specified by a signature $\Sigma$ for monads, this reduction monad $R$ must be equipped with an action of $\Sigma$, hence the following definition:

\begin{definition}
Let $\Sigma$ be a signature for monads. 
The \textbf{category $\GMon^\Sigma$ of reduction  $\Sigma$-models}  is 
defined as the following pullback:
\[
\xymatrix{
\GMon^\Sigma
\ar[r]\ar[d]
\pullback{}
&
\GMon \ar[d]
\\
\Mon^\Sigma
\ar[r]
&
\Mon
}
\]

More concretely, 
\begin{itemize}
    \item a \textbf{reduction $\Sigma$-model} is 
   a reduction monad $R$ equipped with an \textbf{action} $\rho$ of $\Sigma$ in $R$, thus inducing a $\Sigma$-monad that we denote also by $R$, or by $\monadfromred{R}$ when we want to be explicit;
   \item a \textbf{morphism of reduction  $\Sigma$-models} $R\to S$
  is a morphism $f : R\to S$ of reduction monads compatible with the action of $\Sigma$, i.e, whose underlying monad morphism is
  a morphism of models of $\Sigma$.
\end{itemize}
\end{definition}

\subsection{Action of a Reduction Rule}
\label{ss:action-reduction-rule}
Let $\Sigma$ be a signature for monads.
In this section, we introduce the notion of \emph{action
of a reduction rule over $\Sigma$ in a reduction $\Sigma$-model}.
Intuitively, such an action is a ``map from the hypotheses to the conclusion''
of the reduction rule.
To make this precise, we need to first take the product of the hypotheses;
this product is, more correctly, a \emph{fibered} product.

\begin{definition}
Let $(n,p)$ be a term-pair from a $\Sigma$-module $\mathcal{V}$, and $R$ be a reduction $\Sigma$-model.
We denote by $p^*(\Redof{R}^{(n)})$
     the pullback of $\redof{R}^{(n)}:\Redof{R}^{(n)}\to R^{(n)}\times R^{(n)}$ along $p_R:\mathcal{V}(R)\to R^{(n)}\times R^{(n)}$:
     \[
     \xymatrix{
     p^*(\Redof{R}^{(n)})\pullback{}
     \ar[r]
     \ar[d]
     &
     \Redof{R}^{(n)}
     \ar[d]^{\redof{R}^{(n)}}
     \\
     \mathcal{V}(R)\ar[r]_{p_R}
     &
     R^{(n)}
     \times
     R^{(n)}
     }
     \]
     We denote by $p^*(\redof{R}^{(n)}):p^*(\Redof{R}^{(n)}) \to \mathcal{V}(R)$ the projection morphism on the left.
     \end{definition}
     \begin{definition}
     \label{d:mod-hyp-concl}
     Let $\mathcal{A} = (\mathcal{V},(n_i,\phyp_i)_{i\in I}, (n,\pcon))$ be a reduction rule,
      and $R$ be a reduction $\Sigma$-model.
The \textbf{$R$-module $\Hyp_\mathcal{A}(R)$ of hypotheses of $\mathcal{A}$} 
is $\displaystyle{\prod_{i \in I}}{\phantom{\Big\|\!\!}}_{\mathcal{V}(R)} \phyp_i^*{\Redof{R}^{(n_i)}}$, i.e., the fiber product of 
  all the $R$-modules $\phyp_i^*{\Redof{R}^{(n_i)}}$ along their projection to $\mathcal{V}(R)$. It thus comes with a projection $\hyp_\mathcal{A}(R):\Hyp_\mathcal{A}(R)\to\mathcal{V}(R)$
  
  The \textbf{$R$-module $\Concl_\mathcal{A}(R)$ of conclusion of $\mathcal{A}$} 
is $\pcon^*{\Red(R)^{(n)}}$, and comes with a projection $\concl_\mathcal{A}(R):\Concl_\mathcal{A}(R)\to \mathcal{V}(R)$.
\end{definition}
\begin{example}
\label{ex:app-congr-pullbacks}
Let $R$ be a reduction $\Sigma_\LC$-model.
The $R$-module of conclusion of the congruence reduction rule for application (Example~\ref{ex:congr-app-rule}, see also Section~\ref{ss:ex-congr-rule-app}) maps a set $X$ to the
  set of quintuples $(T,T',U,U',m)$ where $(T,T',U,U') \in R^4(X)$ and $m$ is a reduction $m:\redfib{\app(T,U)}{\app(T',U')}$.
  The $R$-module of hypotheses of this reduction rule
  maps a set $X$ to the set of sextuples $(T,T',U,U',m,n)$
  where $(T,T',U,U')\in R^4(X)$, $m : \redfib{T}{T'}$, and $n  : \redfib{U}{U'}$.
\end{example}

\begin{definition} 
\label{d:action-red-rule}
   Let $\mathcal{A}$ be a
   reduction rule over $\Sigma$. 
  An \textbf{action of $\mathcal{A}$ in a reduction  $\Sigma$-model $R$}
  is a morphism between $\hyp_\mathcal{A}(R)$
  and $\concl_\mathcal{A}(R)$ in the slice category $\Mod(R)/\MVar_\mathcal{A}(R)$, that is,
   a morphism of $R$-modules 
  \[
  \tau:\Hyp_\mathcal{A}(R) \to \Concl_\mathcal{A}(R)
  \]
  making the following diagram commute:
    \begin{equation}
    \label{eq:def-action-pbs}
\xymatrix{
\Hyp_\mathcal{A}(R)
\ar[rr]^\tau\ar[rd]
& &
 \Concl_\mathcal{A}(R)
\ar[ld]
\\
&
\MVar_\mathcal{A}(R)
}
\end{equation}
\end{definition}

\begin{example}[Action of the congruence rule for application]
\label{ex:action-app-congr}
Consider the reduction rule of the congruence for application of Example~\ref{ex:congr-app-rule}.
Let $R$ be a reduction $\Sigma_\LC$-model $R$. 
An action $\tau$ in $R$ is an $R$-module morphism which, for each set $X$, maps a sextuple $(T,T',U,U',r,s)$ with $T,T',U,U' \in R(X)$, $r:\redfib{T}{T'}$,
and $s:\redfib{U}{U'}$ to a quintuple $(T,T',U,U',m)$ with $m:\redfib{\app(T,U)}{\app(T',U')}$.
The fact that $\tau$ is fully determined by its last component allows us to present it as in triangle $\eqref{eq:def-action-pbs}$.
 
Alternatively (as justified formally by Lemma~\ref{l:alt-def-action}), an action is a morphism mapping the same sextuple to a reduction $m:\redfib{\app(T,U)}{\app(T',U')}$.
\end{example}

\subsection{Protocol for Specifying Reduction Rules}
\label{ss:protocole-reduction-rules}
In Section~\ref{ss:ex-reduction-rules},
we adopt the following schematic presentation of a reduction rule over a signature $\Sigma$:
    \begin{prooftree}
     \AxiomC{$\redrule{s_1}{t_1} $} 
     \AxiomC{$\dots $} 
     \AxiomC{$\redrule{s_n}{t_n} $} 
\TrinaryInfC{$\redrule{s_0}{t_0} $} 
\end{prooftree}
where $s_i$ and $t_i$ are expressions depending on \enquote{metavariables} $X_1$, \dots, $X_q$. Each pair $(s_i,t_i)$ defines a term-pair as follows:
\begin{align}
p_i & : M_1\times\dots\times M_q 
\to \Theta^{(m_i)}\times  \Theta^{(m_i)}
\notag
\\
p_{i,R,X}(T_1,\dots,T_q) &
:= (s_i[\vec{X}/\vec{T}], t_i[\vec{X}/\vec{T}])
\label{eq:def-p-i}
\end{align}
where $s_i[\vec{X}/\vec{T}]$ is $s_i$ where each
metavariable $X_i$ has been replaced with $T_i$.
 The $\Sigma$-modules $M_1$, \dots, $M_q$, and the natural numbers $m_0$,
\dots, $m_n$ are inferred for Equation~\eqref{eq:def-p-i} to be well defined for all $i\in\{0,\dots,n\}$.

The induced reduction rule is:
\begin{itemize}
    \item \MMM: 
    the $\Sigma$-module of metavariables is $\mathcal{V}=M_1\times\dots\times
    M_q$;  
    \item \HHH: the hypotheses are the term-pairs $(m_i,p_i)_{i\in \{1,\dots,n\}}$;
    \item \CCC: the conclusion is the term-pair $(m_0,p_0)$.
\end{itemize}
Typically, $M_i = \Theta^{(n_i)}$ for some natural number $n_i$, as in the
examples that we consider in this section.
In practice, there are several choices for building the reduction rule out of such a schematic presentation, depending on the order in which the metavariables are picked. 
This order is irrelevant: the different possible versions of reduction rules are all equivalent, in the sense that taking one or the other as part of a reduction signature yields isomorphic categories of models.

\subsection{Examples of Reduction Rules}
\label{ss:ex-reduction-rules}

This section collects a list of motivating examples of reduction rules.

For the rest of this section, we assume that we have fixed a signature
for monads $\Sigma$.  Figure~\ref{fig:ex-reds} shows some notable
examples of reduction rules.  In order, they are: reflexivity,
transitivity, congruence for $\abs$, $\beta$-reduction,
$\eta$-expansion, and expansion of the fixpoint operator.
\begin{figure}
  \begin{equation*}
    \AxiomC{}
    \RightLabel{\scriptsize{Refl}}
    \UnaryInfC{$\redrule{T}{T}$}
    \DisplayProof
    \qquad
    \AxiomC{$\redrule{T}{U}$}
    \AxiomC{$\redrule{U}{W}$}
    \RightLabel{\scriptsize{Trans}}
    \BinaryInfC{$\redrule{T}{W}$}
    \DisplayProof
  \end{equation*}
  \bigskip
  \begin{equation*}
    \AxiomC{$\redrule{T}{U}$}
    \RightLabel{\scriptsize{$\abs$-Cong}}
    \UnaryInfC{$\redrule{\abs(T)}{\abs(U)}$}
    \DisplayProof
    \qquad
    \AxiomC{$\redrule{T}{T'}$}
    \AxiomC{$\redrule{U}{U'}$}
    \RightLabel{\scriptsize{$\app$lr-Cong}}
    \BinaryInfC{$\redrule{\app(T,U)}{\app(T',U')}$}
    \DisplayProof
  \end{equation*}
  \bigskip
  \begin{equation*}
    \AxiomC{$\redrule{T}{T'}$}
    \RightLabel{\scriptsize{$\app$l-Cong}}
    \UnaryInfC{$\redrule{\app(T, U)}{\app(T',U)}$}
    \DisplayProof
    \qquad
    \AxiomC{$\redrule{U}{U'}$}
    \RightLabel{\scriptsize{$\app$r-Cong}}
    \UnaryInfC{$\redrule{\app(T, U)}{\app(T,U')}$}
    \DisplayProof
  \end{equation*}
  \bigskip
  \begin{equation*}
    \AxiomC{}
    \RightLabel{\scriptsize{$\beta$-Red}}
    \UnaryInfC{$\redrule{\app(\abs(T),U)}{T\monsubst{*:=U}}$}
    \DisplayProof
    \quad
    \AxiomC{}
    \RightLabel{\scriptsize{$\mathsf{fix}$-Exp}}
    \UnaryInfC{$\redrule{\fix(T)}{T\monsubst{*:=\fix(T)}}$}
    \DisplayProof
  \end{equation*}
  \bigskip
  \begin{equation*}
    \AxiomC{}
    \RightLabel{\scriptsize{$\eta$-Exp}}
    \UnaryInfC{$\redrule{T}{\abs(\app(\iota(T),*))}$}
    \DisplayProof
    \quad
    \AxiomC{}
    \RightLabel{\scriptsize{$\eta$-Contr}}
    \UnaryInfC{$\redrule{\abs(\app(\iota(T),*))}{T}$}
    \DisplayProof
  \end{equation*}

  \bigskip
Here, $\iota : \Theta \to \Theta'$ denotes the canonical morphism $\iota_{R,X} : R(X)
\to R(X+1)$.
\caption{Examples of reduction rules.}
\Description{The figure specifies various reduction rules in our framework.}
  \label{fig:ex-reds}
\end{figure}

For the example of the fixpoint operator (rule $\mathsf{fix}$-Exp), we consider the signature $\Sigma_\fix$, as described in \cite[Section 6.4]{FSCD2019}
(but without enforcing the fixpoint equation, which is replaced here by the reduction rule under consideration).
A model of $\Sigma_\fix$ is a monad $R$ equipped with an $R$-module morphism $\fix : R' \to R$.

Figure~\ref{fig:ex-reds-details} lists the modules and term pairs for
hypotheses and conclusion of each of these reduction
rules. There, $\pi_{i,j}$ designates the pair projection described in Definition~\ref{d:pair-projection}.
\begin{figure}
  \centering
  \begin{tabular}{l|c|c|c|c}
    Rule & Signature 
    & Metavariables & Hypotheses & Conclusion\\
    \hline
    Refl & any & $\Theta$ && $(0, \langle\id,\id\rangle)$ \\
    Trans & any & $\Theta^3$ for $(T,U,W)$&  $(0,\pi_{1,2})$, $(0,\pi_{2,3})$  & $(0, \pi_{1,3})$ \\
    $\abs$-Cong & $\Sigma_\LC$ & $\Theta'\times\Theta'$ for $(T,U)$ & $(1,\id)$ & $(0, \abs\times\abs)$ \\
    $\app$lr-Cong & $\Sigma_\LC$ & $\Theta^4$ for $(T,T',U,U')$
                    & $(0,\pi_{1,2}), (0,\pi_{3,4})$
                                 & $(0, \app\times\app)$ \\
    $\app$l-Cong & $\Sigma_\LC$ & $\Theta^3$ for $(T,T',U)$ & $(0,\pi_{1,2})$
                                 & $(0, \langle \app\circ\pi_{1,3} , \app\circ\pi_{2,3} \rangle)$ \\
    $\app$r-Cong & $\Sigma_\LC$ & $\Theta^3$ for $(T,U,U')$ & $(0,\pi_{2,3})$
                                 & $(0, \langle \app\circ\pi_{1,2} , \app\circ\pi_{1,3} \rangle)$ \\
    $\beta$-Red & $\Sigma_\LC$ & $\Theta'\times\Theta$ for $(T,U)$ & &
                       $(0, c_{\beta\text{-Red}})$
    \\
    $\fix$-Exp & $\Sigma_\fix$ & $\Theta'$ & & $(0, c_{\fix\text{-Exp}})$ \\
    $\eta$-Exp & $\Sigma_\LC$ & $\Theta$ & & $(0, \langle \id, b_\eta \rangle )$ \\
    $\eta$-Contr & $\Sigma_\LC$ & $\Theta$ & & $(0, \langle b_\eta, \id\rangle)$
  \end{tabular}

  \begin{align*}
    c_{\beta\text{-Red},R,X}(T,U) & = \langle\app(\abs(T),U) , T\monsubst{*:= U}\rangle \\
    c_{\fix\text{-Exp},R,X}(T) & = \langle\fix(T) , T\monsubst{*:= \fix(T)}\rangle \\
    b_\eta(T) & = \abs (\app (\iota\,T), *)
\end{align*}
\caption{Modules and term pairs relative to the
    reduction rules of Figure~\ref{fig:ex-reds}. }
\label{fig:ex-reds-details}
\Description{The figure specifies the modules and term pairs of the reduction rules given in the previous figure.}
\end{figure}
Below we present different sets of reduction rules over the signature
$\Sigma_\LC$  for different variants of lambda calculus. 

\begin{center}
  \begin{tabular}{c|l}
    Variant of lambda calculus & Associated reduction rules \\
    \hline
    weak head $\beta$ (Example~\ref{ex:lambda-whb})& $\beta$-Red, $\app$l-Cong \\
    congruent $\beta$ (Example~\ref{ex:lc-congruent-beta}) & $\beta$-Red, $\abs$-Cong, $\app$l-Cong, $\app$r-Cong \\
    parallel $\beta$ (Example~\ref{ex:lc-parbeta}) & $\beta$-Red, $\abs$-Cong, $\app$lr-Cong 
  \end{tabular}
\end{center}

\section{Signatures for reduction monads and Initiality}
\label{sec:signatures}
\label{sec:initiality}
In this section, we define the notion of \emph{reduction signature}, consisting of a signature for monads $\Sigma$ and a family of reduction rules over $\Sigma$ (see Section~\ref{ss:signatures-models}).
As usual, we assign to each such signature a \emph{category of models}.
We call a reduction signature \emph{effective} if the associated category of models has an initial object.
Our main result, Theorem~\ref{thm:alg-2-sig-initial} (see Section~\ref{ss:main-result}), states that a reduction signature is effective as soon as its underlying signature for monads is effective.

\subsection{Signatures and their Models}
\label{ss:signatures-models}

We define here \emph{reduction signatures} and their \emph{models}.

\begin{definition}
  A \textbf{reduction signature} is a pair $(\Sigma,\mathfrak{R})$ of a signature $\Sigma$ for monads and a
  family $\mathfrak{R}$ of reduction rules over ${\Sigma}$. 
\end{definition}

\begin{definition}
  Given a reduction monad $R$ and a reduction signature $\mathcal{S}=(\Sigma,\mathfrak{R})$, an \textbf{action
  of $\mathcal{S}$ in $R$} consists of an action of $\Sigma$ in 
  its underlying monad $\monadfromred{R}$
  and an action of each reduction rule of $\mathfrak{R}$ in $R$.
\end{definition}

\begin{definition}
Let $\mathcal{S}=(\Sigma,\mathfrak{R})$ be a reduction signature. A \textbf{model of $\mathcal{S}$} 
is a 
reduction monad equipped with an action of $\mathcal{S}$, or equivalently, a reduction $\Sigma$-model equipped with an action of each reduction rule of $\mathfrak{R}$.
\end{definition}

\subsection{The Functors \texorpdfstring{$\Hyp_\mathcal{A}$}{Hyp A} and \texorpdfstring{$\Concl_\mathcal{A}$}{Concl A}}
\label{ss:hyp-concl-functors}
The definition of morphism between models of a reduction signature relies on the functoriality of the assignments $R\mapsto \Hyp_\mathcal{A}(R)$
and $R\mapsto \Concl_\mathcal{A}(R)$, for a given reduction rule $\mathcal{A}$ on a signature $\Sigma$ for monads.
\begin{definition}
\label{d:red-sigma-module}
Let $\Sigma$ be a signature for monads, and  $\mathcal{A}$ be a reduction rule
over $\Sigma$. Definition~\ref{d:mod-hyp-concl} assigns to each model $R$ of $\Sigma$ the $R$-modules $\Hyp_\mathcal{A}(R)$ and $\Concl_\mathcal{A}(R)$. These assignments extend to functors $\Hyp_\mathcal{A},\Concl_\mathcal{A}:\GMon^\Sigma\to\LMod$.
\end{definition}
\begin{proposition}
\label{prop:hyp-concl-preserve-mon}
Given the same data, the functors 
$\Hyp_\mathcal{A}$ and $\Concl_\mathcal{A}$ 
are $\GMon^\Sigma$-modules, i.e., they
 commute with the forgetful functors to $\Mon$:
\[
  \begin{xy}
   \xymatrix{
                    \GMon^\Sigma \ar[rd] \ar@<+.5ex>[rr]^{\Hyp_\mathcal{A}}
                    \ar@<-.5ex>[rr]_{\Concl_\mathcal{A}} 
                    &        &   \LMod \ar[dl] \
                    \\
                                     & \Mon
   }
  \end{xy}
 \]
\end{proposition}
\subsection{The Main Result}
\label{ss:main-result}
\label{ss:category-models}

For a reduction signature $\mathcal{S}$, we define here the notion of $\mathcal{S}$-model morphism, inducing a \textbf{category of models of $\mathcal{S}$}.
We then state our main result, Theorem~\ref{thm:alg-2-sig-initial}, which gives
a sufficient condition for $\mathcal{S}$  to admit an initial model.

\begin{definition}
\label{d:model-mor}
Let $\mathcal{S}=(\Sigma,\mathfrak{R})$ be a reduction signature.
A morphism between models $R$ and $T$ of $\mathcal{S}$ is a morphism $f$ of
reduction $\Sigma$-models commuting with the action of any reduction rule
$\mathcal{A}\in \mathfrak{R}$, i.e., such that, for any such $\mathcal{A}$, the
following diagram of natural transformations commutes:
\[
  \xymatrix{
  \Hyp_\mathcal{A}(R) \ar[r] \ar[d]_{\Hyp_\mathcal{A}(f)} & \Concl_\mathcal{A}(R)
  \ar[d]^{\Concl_\mathcal{A}(f)}
  \\
  \Hyp_\mathcal{A}(T) \ar[r] & \Concl_\mathcal{A}(T)
  }
  \]
\end{definition}

\begin{example}[Example~\ref{ex:action-app-congr} continued]
 Consider the reduction signature consisting of the signature $\Sigma_\app$ of a binary
 operation $\app$ and the reduction rule of congruence (Section~\ref{ss:ex-congr-rule-app}) for application (Example~\ref{ex:congr-app-rule}).
 
 Let $R$ and $T$ be models for this signature: they are reduction
 $\Sigma_\app$-models equipped with an action $\rho$ and $\tau$, in the
 alternative sense of Example~\ref{ex:action-app-congr}. A reduction $\Sigma_\app$-model morphism $(f,\alpha)$ between $R$ and $T$ is a model morphism if, for any set $X$, any sextuple
 $(A,A',B,B',m,n)$
  with $(A,A',B,B')\in R^4(X)$, $m : \redfib{A}{A'}$, and $n  : \redfib{B}{B'}$, the reduction $\rho(A,A',B,B'):\redfib{\app(A,B)}{\app(A',B')}$ is mapped to the reduction $\tau(f(A),f(A'),f(B),f(B'))$ by $\alpha : \Redof{R}\to\Redof{T}$.
\end{example}

\begin{proposition}
  Let
  $\mathcal{S}=(\Sigma,\mathfrak{R})$ be a reduction signature. Models of $\mathcal{S}$ and their morphisms, with the obvious composition and identity, define a category that we denote by  $\GMon^{\mathcal{S}}$, equipped with a forgetful functor to $\GMon^\Sigma$.
\end{proposition}

\begin{definition}
  A reduction signature $\mathcal{S}$ is said to be
  \textbf{effective}
  if its category of models $\GMon^{\mathcal{S}}$ has an initial object,
  denoted $\widehat{\mathcal{S}}$.
  In this case, we say that $\widehat{\mathcal{S}}$
  (or more precisely the underlying reduction monad)
  is \textbf{generated (or specified) by $\mathcal{S}$}.
\end{definition}

We now have all the ingredients required to state our main result:

\begin{theorem}
\label{thm:alg-2-sig-initial}
Let $(\Sigma,\mathfrak{R})$ be a reduction signature. 
If $\Sigma$ is effective, then so is $(\Sigma,\mathfrak{R})$.
\end{theorem}
The proof of this theorem is given in Section~\ref{s:effective-alg-2sig}.
\begin{definition}
A reduction signature $(\Sigma,\mathfrak{R})$ is called \textbf{algebraic} if
$\Sigma$ is (in the sense of \cite{FSCD2019}).
\end{definition}
Theorem~32 of \cite{FSCD2019}
entails the following corollary:
\begin{corollary}
\label{cor:effective-if-2-sig-alg}
Any algebraic reduction signature is effective.
\end{corollary}
All the examples of reduction signatures considered here
satisfy the condition of Corollary~\ref{cor:effective-if-2-sig-alg}.
\begin{example}[Variants of lambda calculus]
  \label{ex:variants-lc}
 In Section~\ref{ss:ex-reduction-rules}, we considered different sets of
 reduction rules for different variants of lambda calculus. Each such set
 $\mathfrak{R}$ defines an algebraic reduction signature $(\Sigma_\LC,\mathfrak{R})$.
\begin{center}
  \begin{tabular}{c|c|c}
    Variant of lambda calculus &   Signature & Initial model \\
    \hline
    weak head $\beta$ (Example~\ref{ex:lambda-whb})&  $\mathcal{S}_{\LCwhb}$ & $\LCwhb$ \\
    congruent $\beta$ (Example~\ref{ex:lc-congruent-beta}) & 
                                                             $\mathcal{S}_{\LCb}$ & $\LCb$ \\
    parallel $\beta$ (Example~\ref{ex:lc-parbeta})
                               & 
$\mathcal{S}_{\LCbpar}$ & $\LCbpar$ 
  \end{tabular}
\end{center}
\end{example}

\begin{rem}[Continuation of Remark~\ref{r:red-mon-proof-relevant}]
 \label{rem:signatures-proof-relevant}
 Just as our reduction monads are ``proof-relevant'' (cf.\ Remark~\ref{r:red-mon-proof-relevant}),
 our notion of reduction signature allows for the specification of multiple reductions between terms. As a trivial example , duplicating the $\beta$-rule in the signature $\mathcal{S}_{\LCb}$ yields two distinct $\beta$-reductions in the initial model.
\end{rem}

\begin{example}[Reduction signature of lambda calculus with a fixpoint operator]
  \label{ex:redsig-lc-fixpoint}
  The signature $\mathcal{S}_\LCfix$ specifying the reduction monad $\LCfix$ of the lambda calculus with a fixpoint operator 
  extends the signature $\mathcal{S}_{\LCb}$ of Example~\ref{ex:variants-lc}
  with:
  \begin{itemize}
    \item a new operation $\fix:\Theta'\to \Theta$ (thus extending the signature for
      monads $\Sigma_\LC$);
      \item the reduction rule for the fixpoint reduction (cf.\ Section~\ref{ss:ex-reduction-rules});
      \item a congruence rule for $\fix$:
        \[
          \frac{
           \redrule{T}{T'}
          }
          {
            \redrule{\fix(T)}{\fix(T')}
          }
        \]
    \end{itemize}
 \end{example}

\section{Proof of Theorem~\ref{thm:alg-2-sig-initial}}
\label{s:effective-alg-2sig}

This section details the proof of Theorem~\ref{thm:alg-2-sig-initial}.

Let $\mathcal{S}=(\Sigma , (\mathcal{A}_i)_{i\in I}))$ be a reduction signature.
We denote by $\mathcal{U}^\Sigma$ the forgetful functor from the category of
reduction $\Sigma$-models to the category of models of $\Sigma$.

In Section~\ref{ss:normalizing}, we first reduce to the case of reduction rules
$(\mathcal{V},(n_j,\phyp_j)_{j\in J},(n,\pcon))$
for which $n=0$, that we call \emph{normalized}. Then, in Section~\ref{ss:alt-def-models}, we give an alternative definition of the category of models that we make use of in the proof of effectivity, in Section~\ref{ss:initiality}.

\subsection{Normalizing Reduction Rules}
\label{ss:normalizing}
\begin{definition}
A reduction rule $(\mathcal{V},(n_j,\phyp_j)_{j\in J},(n,\pcon))$ is said to be \textbf{normalized} if $n=0$.
\end{definition}

\begin{lemma}
\label{l:normalized-red-rule}
Let $\mathcal{A}=(\mathcal{V},(n_j,\phyp_j)_{j\in J},(n,\pcon))$  be a reduction rule over $\Sigma$.
Then there exists a normalized reduction rule  $\mathcal{A}'$ over $\Sigma$ such that the induced notion of action is equivalent, in the sense that:
\begin{itemize}
    \item for a reduction $\Sigma$-model $R$, there is a bijection between actions of $\mathcal{A}$ in $R$ and actions of $\mathcal{A}'$ in $R$;
    \item a morphism between reduction $\Sigma$-models equipped with an action of $\mathcal{A}$ preserves the action (in the sense of Definition~\ref{d:model-mor}) if and only if it preserves the corresponding action of $\mathcal{A}'$ through the bijection.
\end{itemize}
\end{lemma}
Before tackling the proof, we give an alternative definition of action and model morphism:
\begin{lemma}
\label{l:alt-def-action}
Let $\mathcal{A}=(\mathcal{V}, (n_i,\phyp_i)_{i\in I},(n,\pcon))$ be a
   reduction rule over $\Sigma$. 
   By universal property of the pullback $\Concl_\mathcal{A}(R)=\pcon^*\Redof{R}^{(n)}$, an action
   can be alternatively be defined as an $R$-module morphism $\sigma : \Hyp_\mathcal{A}(R)\to \Redof{R}^{(n)}$ making the following diagram commute
    \begin{equation}
    \label{eq:action-alt-def}
    \xymatrix{
\Hyp_\mathcal{A}(R)
 \ar[d] \ar[r]^\sigma & {\Redof R}^{(n)}\ar[d]^{{\redof R}^{(n)}} \\
\mathcal{V}(R)\ar[r]_{\pcon} &  R^{(n)}\times R^{(n)}
}
\end{equation}
\end{lemma}
\begin{lemma}
\label{l:alt-commut-action}
Using this alternative definition of action,
a morphism between models $R$ and $T$ of a reduction signature $\mathcal{S}=(\Sigma,\mathfrak{R})$ is a morphism $f$ of reduction $\Sigma$-models making the following diagram commute, for any reduction rule $\mathcal{A}=(\mathcal{V}, (n_i,l_i,r_i)_{i\in I},(n,l,r))$ of $ \mathfrak{R}$:
\[
  \xymatrix{
  \Hyp_\mathcal{A}(R) \ar[r] \ar[d]_{\Hyp_\mathcal{A}(f)} & \Redof{R}^{(n)}\ar[d]^{\Redof{f}^{(n)}}
  \\
  \Hyp_\mathcal{A}(T) \ar[r] & \Redof{T}^{(n)}
  }
  \]
\end{lemma}
We now prove Lemma~\ref{l:normalized-red-rule} using these alternative definitions:
\begin{proof}[Proof of Lemma \ref{l:normalized-red-rule}]
The reduction rule $\mathcal{A}'=(\mathcal{V}',(n_j,\phyp_j')_{j\in J},(0,\pcon'))$ is defined as follows:
\begin{itemize}
    \item \MMM:
    $\mathcal{V}' = \mathcal{V}\times \Theta^n$
    \item \HHH:
    For each $j\in J$,
    $\phyp_j':\mathcal{V}'\to \Theta^{(n_j)}\times\Theta^{(n_j)}$ is defined as the composition of $\pi_1 : \mathcal{V}\times \Theta^n\to\mathcal{V}$ with $\phyp_j:\mathcal{V}\to \Theta^{(n_j)}\times\Theta^{(n_j)}$.
    \item \CCC:
    The morphism     $\pcon':\mathcal{V}\times\Theta^n \to \Theta\times\Theta$
    maps a model $R$ of $\Sigma$ to the $n^{th}$ transpose of $\pcon:\mathcal{V}(R) \to R^{(n)}\times R^{(n)}$ with respect to the adjunction $\_\times R \dashv \_'$ in $\Mod(R)$ described in \cite[Proposition 13]{ahrens_et_al:LIPIcs:2018:9671}.
\end{itemize}
Now, consider an action for the reduction rule $\mathcal{A}$ in a reduction $\Sigma$-model $R$: it is an 
$R$-module morphism $\tau :\Hyp_\mathcal{A}(R) \to \Redof{R}^{(n)}$ such that the following square commutes:
\[
\xymatrix{
\Hyp_\mathcal{A}(R) \ar[d] \ar[rr]^\tau && \Redof{R}^{(n)}\ar[d]^{\redof{R}^{(n)}} \\
\mathcal{V}(R)\ar[rr]_{\pcon} && R^{(n)}\times R^{(n)}
}
\]
Equivalently, through the adjunction mentioned above, it is given by an $R$-module morphism $\tau^* : \Hyp_R\times R^m \to M$ such that the following diagram commutes:
\[
\xymatrix{
\Hyp_\mathcal{A}(R) \times R^n \ar[d] \ar[rr]^{\tau^*} && \Redof{R} \ar[d]^{\redof{R}} \\
\mathcal{V}(R)\times R^n\ar[rr]_{\pcon^*} && R \times R
}
\]
This is exactly the definition of an action of $\mathcal{A}'$. It is then straightforward to check that one action is preserved by a reduction monad morphism if and only if the other one is.
\end{proof}

\begin{corollary}
For each reduction signature, there exists a reduction signature yielding an isomorphic category of models and whose underlying reduction rules are all normalized.
\end{corollary}
\begin{proof}
Just replace each reduction rule with the one given by Lemma~\ref{l:normalized-red-rule}.
\end{proof}
Thanks to this lemma, we assume in the following that all the reduction rules of the given signature $\mathcal{S}$ are normalized. 
\subsection{Models as Vertical Algebras}
\label{ss:alt-def-models}
In this section, we give an alternative definition for the category of models of $\mathcal{S}$ that is convenient in the proof of effectivity.

First we rephrase the notion of action of a reduction rule as an algebra structure for a suitably chosen endofunctor.
Indeed, an action of a normalized reduction rule $\mathcal A=(\mathcal{V},(n_j,\phyp_j)_{j\in J},(0,\pcon))$ in a reduction  $\Sigma$-model $R$ is given by a $R$-module morphism $\tau :\Hyp_\mathcal{A}(R) \to \Redof{R}$ such that the following square commutes:
\[
\xymatrix{
\Hyp_\mathcal{A}(R) \ar[d] \ar[rr]^\tau && \Redof{R}\ar[d]^{\redof{R}} \\
\mathcal{V}(R)\ar[rr]_{p} && R\times R
}
\]
We can rephrase this commutation by stating that this morphism $\tau$ is a morphism in the slice category $\Mod(\monadfromred{R})/\monadfromred{R}^2$ from an object that we denote by $F_{\mathcal{A}|\monadfromred{R}}(\Redof{R},\redof{R})$, to $(\Redof{R},\redof{R})$.
Actually, the domain is functorial in its argument, and thus
the action $\tau$ can be thought of as an algebra structure on $(\Redof{R},\redof{R})$:

\begin{lemma}
\label{l:def-alt-action}
Given any model $R$ of $\Sigma$,
the assignment $(M,p:M\to R\times R) \mapsto F_{\mathcal{A}|R}(M,\pcon)$ yields an endofunctor $F_{\mathcal{A}|R}$ on $\Mod(R)/R^2$. An action of $\mathcal{A}$ in a reduction $\Sigma$-model $R$ is exactly the same as an algebra structure for this endofunctor on $(\Redof{R},\redof{R})\in \Mod(R)/R^2$.

Furthermore, the assignment $R\mapsto
F_{\mathcal{A}|\monadfromred{R}}(\Redof{R},\redof{R})$ yields an endofunctor
$F_\mathcal{A}$ on the category of reduction $\Sigma$-models. This functor
preserves the underlying model of $\Sigma$, in the sense that
 $\mathcal{U}^\Sigma\cdot F_\mathcal{A}=\mathcal{U}^\Sigma$.
 \end{lemma}
 \begin{proof}
 This is a consequence of the functoriality of $\Hyp_\mathcal{A}$, as noticed in Section~\ref{ss:hyp-concl-functors}.
 \end{proof}
 
 Now, we give our alternative definition of the category of models:
 \begin{proposition}
\label{prop:models-vert-alg}
Let $F_\mathcal{S} : \GMon^\Sigma \to \GMon^\Sigma$ be the coproduct $\coprod_i F_{\mathcal{A}_i}$. Then, the category of models of $\mathcal{S}$ is isomorphic to the \textbf{category of vertical algebras} of $F_\mathcal{S} $ defined as follows:
\begin{itemize}
    \item an object is an algebra $r:F_{\mathcal{S}}(R)\to R$ such that  $r$ is mapped to the identity by $\mathcal{U}^\Sigma$
    \item morphisms are the usual $F_\mathcal{S}$-algebra morphisms.
\end{itemize}
\end{proposition}

We adopt this definition in the following. We show now a property of the category of models that will prove useful in the proof of effectivity:
\begin{lemma}
\label{prop:model-model-fibration}
The forgetful functor from the category of models of $\mathcal{S}$ to the category of   
 models of $\Sigma$ is a fibration.
\end{lemma}
The proof relies on some additional lemmas, in particular the following one, that we will specialize by taking $p=\mathcal{U}^\Sigma$ 
(requiring to show that $\mathcal{U}^\Sigma:\GMon^\Sigma\to \Mon^\Sigma$ is a fibration)
and $F=F_\mathcal{S}$:
\begin{lemma}
\label{l:vert-alg-fibration}
  Let $p: E \to B$ be a fibration and $F$ and endofunctor on $E$ satisfying $p\cdot F = p$.
  Then the category of vertical algebras of $F$ is fibered over $B$.
\end{lemma}
\begin{proof}

Let $r : F(R)\to R$ be an algebra over $X \in B$. Let $a:Y\to X$ be a morphism
in $B$. Let $\overline{a} : a^* R \to R$ be the associated cartesian morphism in $E$.
We define the reindexing of $r$ along $a$ as follows: the base object is $a^* R$, and the algebra structure 
$\rho:F(a^* R)\to R$ is given by the unique morphism which factors 
$\xymatrix{F(a^* R) \ar[r]^{F(\overline{a})} & F(R) \ar[r]^r & R}$
through the cartesian morphism $\overline{a} : a^* R \to R$. Thus, the square
\[
\xymatrix{
F(a^*R) \ar[r]^{F(\overline{a})} 
\ar[d]_\rho
&
F(R)
\ar[d]^r
\\
a^* R \ar[r]_{\overline{a}}
&
R
}
\]
commutes, so $\overline{a}$ is a morphism of algebras between $\rho$
and $r$. 
Next, we prove that it is a cartesian morphism: let $s:F(S)\to S$ be a vertical algebra over an object $Z$ of $B$, and $v:s\to r$ be a morphism of algebras such that there exists $b:Z\to Y$ such that $p(v)=\xymatrix{Z \ar[r]^b & Y \ar[r]^a & X}$. 
We need to show that there exists a unique algebra morphism $w:s\to \rho$ such that $v=\overline{a}\circ w$ and $p(w)=b$. Uniqueness follows from the fact that $\overline{a}$ is cartesian for the fibration $p:E\to B$.
Moreover, as $\overline{a}$ is cartesian, we get a morphism $w:S\to a^* R$. We turn it into an algebra morphism by showing that the following square commutes:
\[
\xymatrix{
F(S)\ar[r]^{F(w)}\ar[d]_s & F(a^*R)\ar[d]^\rho
\\
S \ar[r]_w &  a^*R
}
\]
As $\overline{a}$ is cartesian and both $w$ and $F(w)$ are sent to $b$ by $p$, it is enough to show equalities of both morphisms after postcomposing with $\overline{a}$. This follows from $v$ being an algebra morphism.
\end{proof}
We want to apply this lemma for proving Lemma~\ref{prop:model-model-fibration}. 
We thus need to show 
 that $\mathcal{U}^\Sigma : \GMon^\Sigma \to \Mon^\Sigma$ is a fibration:
\begin{lemma}
\label{l:fibration-U}
\label{l:fibration-trans}
The forgetful functors $\GMon \to \Mon$ and $\mathcal{U}^\Sigma : \GMon^\Sigma \to \Mon^\Sigma$ are fibrations.
\end{lemma}
\begin{proof}

We have the two following pullbacks:
\[
\xymatrix{
\GMon^\Sigma\ar[r]\ar[d]_{\mathcal{U}^\Sigma}\pullback{} &
\GMon\ar[r]\ar[d]\pullback{}
& V(\LMod)\ar[d]^{\codom}
\\
\Mon^\Sigma\ar[r]&
\Mon\ar[r]_{\Theta\times\Theta} &
\LMod
}
\]
where $V(\LMod)$ is the full subcategory of arrows of $\LMod$ which are vertical (that is, they are mapped to the identity monad morphism by the functor from $\LMod$ to $\Mon$),
and $\codom$ maps such an arrow to its codomain. 
By \cite[Propositions 4 and 8]{ahrens_et_al:LIPIcs:2018:9671}, the category $\LMod$ has fibered finite limits, so that 
$\codom$ is a fibration (\cite[Exercise 9.4.2 (i)]{JacobsCLTT}).

Now, Proposition 8.1.15 of \cite{borceux_1994} states that a pullback of a fibration is a fibration. Thus, the middle functor $\GMon\to\Mon$ is a fibration, and then, $\mathcal{U}^\Sigma:\GMon^\Sigma\to \Mon^\Sigma$ also is.
\end{proof}
Finally, gathering all these lemmas yields a proof that the category of models
of $\mathcal{S}$ is indeed fibered over the category of models of $\Sigma$:
\begin{proof}[Proof of Lemma~\ref{prop:model-model-fibration}]
Apply Lemma~\ref{l:vert-alg-fibration} with the fibration $p=\mathcal{U}^\Sigma$ (Lemma~\ref{l:fibration-U}) and $F=F_\mathcal{S}$.
\end{proof}

\subsection{Effectivity}
\label{ss:initiality}
In this section, we prove that $\mathcal{S}$ has an initial model, provided that
there exists an initial model of $\Sigma$.
The category of models of $\mathcal{S}$ is fibered over the category of models
of $\Sigma$. A promising candidate for the initial model is the initial object,
if it exists, in the fiber category over the initial model of $\Sigma$:
\begin{lemma}
\label{lem:initial-fibration}
Let $p:E\to B$ be a fibration, $b_0$ be an initial object in $B$ and $e_0$ be an object over $b_0$ that is initial in the fiber category over $b_0$. Then $e_0$ is initial in $E$.
\end{lemma}
In the following, we thus construct the initial object in a fiber category over
a given model $R$ of $\Sigma$.
This fiber category can be characterized as a category of algebras:
\begin{lemma}
\label{l:fiber-cat-model}
The fiber category over a given model $R$ of $\Sigma$ through the fibration from models of $\mathcal{S}$ (Lemma~\ref{prop:model-model-fibration}) is the category of algebras of the endofunctor $F_{\mathcal{S}|R}=\coprod_i F_{\mathcal{A}_i|R}$ on the slice category $\Mod(R)/R^2$.
\end{lemma}
Thus, our task is to construct the initial algebra of some specific endofunctor.
Adámek's theorem \cite{Adamek1974} provides a sufficient condition for the existence of an initial algebra:
\begin{lemma}[Adámek]
\label{l:adamek}
Let $F$ be a finitary endofunctor on a cocomplete category $C$. Then the category of algebras of $F$ has an initial object.
\end{lemma}
This initial object can be computed as a colimit of a chain, but we do not rely here on the exact underlying construction.

The first requirement to apply this lemma is that the base category is cocomplete, and this is indeed the case:
\begin{lemma}
\label{l:cocomplete-fiber-cat}
The category $\Mod(R)/R^2$ is cocomplete for any monad $R$.
\end{lemma}
\begin{proof}
The category of modules $\Mod(R)$ over a given monad $R$ is cocomplete \cite[Proposition 4]{ahrens_et_al:LIPIcs:2018:9671}, so any of its slice categories is, by the dual of \cite[Exercise V.1.1]{maclane}, in particular $\Mod(R)/R^2$.
\end{proof}

Let us show that the finitarity requirement of Lemma~\ref{l:adamek} is also satisfied for the case of a signature with a single reduction rule:

\begin{lemma}
\label{lem:finitary-elem-arity}
 Let $\mathcal{A}=(\mathcal{V},(n_i,\phyp_i)_{i\in I},(0,\pcon))$ be a normalized reduction rule over $\Sigma$, and
 $R$ be a model of $\Sigma$. Then, $F_{\mathcal{A}|R}$ is finitary.
\end{lemma}
\begin{proof}
In this proof, we denote by $F$ the endofunctor $F_{\mathcal{A}|R}$ on $\Mod(R)/R^2$, 
 by $\pi : D/d\to D$ the projection for a general slice category, and by $\alpha : \pi \to d$ the natural transformation from $\pi$ to the functor constant at $d$ induced by the underlying morphism of a slice object: $\alpha_p:\pi(p)\to d$.
Note that $\pi$ creates colimits, by the dual of \cite[Exercise V.1.1]{maclane}.

Given a filtered diagram we want to show that the image by $F$ of the colimiting cocone is colimiting. As $\pi$ creates colimits, this is enough to show that the image by $\pi\cdot F$ of the colimiting cocone is colimiting. Thus, it is enough to prove that $\pi\cdot F : \Mod(R)/R^2 \to \Mod(R)$ is finitary.

Given any $q \in \Mod(R)/R^2$ the module $\pi( F(q))$ is $\Hyp_\mathcal{A}(R)$, which can be computed as the limit of the following finite diagram:
\[
     \xymatrix{
     & & \mathcal{V}(R)
     \ar[rd]_{\phyp_{i',R}} 
     \ar[ld]_{\phyp_{i,R}}
     \ar[rrd] &
     \\
   &  R^{(n_i)}
     \times R^{(n_i)}
     & &
     R^{(n_{i'})}
     \times R^{(n_{i'})}
     &
     \dots
     \\
  &   \pi(q)^{(n_i)}\ar[u]^{\alpha_q^{(n_i)}}
     &&
     \pi(q)^{(n_{i'})}\ar[u]^{\alpha_q^{(n_{i'})}}
     }
     \]
     
     Let $J:\CC \to \Mod(R)/R^2$ be a filtered diagram. 
     As $\pi$ preserves colimits (since it creates them), $\pi( F(\colim J))$ is the limit of the following diagram:
     \[
     \xymatrix{
     &
     & \mathcal{V}(R)
    \ar[rd]_{\phyp_{i',R}} 
     \ar[ld]_{\phyp_{i,R}}
     \ar[rrd] &
     \\
     & R^{(n_i)}
     \times R^{(n_i)}
     & &
     R^{(n_{i'})}
     \times R^{(n_{i'})}
     &
     \dots
     \\
     &
     \colim \pi(J)^{(n_i)}\ar[u]^{ \alpha_{\colim J}^{(n_i)}}
     &&
     \colim \pi(J)^{(n_{i'})}\ar[u]^{ \alpha_{\colim J}^{(n_{i'})}}
     }
     \]
Now, as limits and colimits are computed pointwise in the category of modules, and as finite limits commute with filtered colimits in $\Set$ (\cite[Section~IX.2, Theorem~1]{maclane}), we have that $\pi( F(\colim J))$, as the limit of such a diagram, is canonically isomorphic to the colimit of $\pi\cdot F\cdot J$. 
\end{proof}

Now, consider a signature $\mathcal{S}$ with a family of reduction rules
$(\mathcal{A}_i)_i$. The functor that we are concerned with is
$F_{\mathcal{S}|R}=\coprod_i F_{\mathcal{A}_i|R}$, for a given model $R$ of $\Sigma$:

\begin{lemma}
\label{l:finitary-big-F}
For any model $R$ of $\Sigma$, the functor
$F_{\mathcal{S}|R}=\coprod_i F_{\mathcal{A}_i|R}$ is finitary.
\end{lemma}
\begin{proof}
This is a coproduct of finitary functors (by Lemma~\ref{lem:finitary-elem-arity}), and so is finitary as
 colimits commute with colimits, by \cite[Equation V.2.2]{maclane}.
\end{proof}

Now we are ready to tackle the proof of our main result:
\begin{proof}[Proof of Theorem~\ref{thm:alg-2-sig-initial}]
We assume that $\Sigma$ is effective; let $R$ be the initial model of $\Sigma$. We want to show that $\mathcal{S}$ has an initial model.
 We apply Lemma~\ref{lem:initial-fibration} with $p$ the fibration from models
 of $\mathcal{S}$ to models of $\Sigma$ (Lemma~\ref{prop:model-model-fibration}): we are left with providing an initial object in the fiber category over $R$.
 By Lemma~\ref{l:fiber-cat-model}, this boils down to constructing an initial algebra for the endofunctor ${F_{\mathcal{S}|R}}$ on the category $\Mod(R)/R^2$. We apply Lemma~\ref{l:adamek}: $\Mod(R)/R^2$ is indeed cocomplete by Lemma~\ref{l:cocomplete-fiber-cat}, and $F_{\mathcal{S}|R}$ is finitary by Lemma~\ref{l:finitary-big-F}).
\end{proof}

\section{Example: Lambda calculus with explicit substitutions}
\label{s:ex-lex}

Here, we give a  signature specifying the reduction monad of the lambda calculus
with explicit substitutions as described in~\cite{lambda-ex}. One feature of this example is that it involves operations subject to some equations, and on top of this syntax with equations, a ``multigraph of reductions''.

In Section~\ref{ss:sig-monad-lambda-ex}, we present the underlying signature for monads, and in Section~\ref{ss:reduction-rules-lambda-ex}, we list the reduction rules of the signature. 

\subsection{Signature for the Monad of the Lambda Calculus with Explicit Substitutions}
We give here the signature for the monad of the lambda calculus with explicit substitutions: first the syntactic operations, and then the equation that the explicit substitution must satisfy.
\label{ss:sig-monad-lambda-ex}
\subsubsection{Operations}
\label{sss:1-sig-lambda-ex}
The lambda calculus with explicit substitutions extends the lambda calculus with an explicit unary substitution operator $t[x/u]$. 
Here, the variable $x$ is assumed not to occur freely in $u$.
In our setting, it is specified
as an operation $\s_X : \LC'(X) \times \LC(X) \to \LC(X)$.
It is thus specified by the signature $\Theta'\times \Theta$. 
An action of this signature in a monad $R$ yields a map $\s_X:R(X+\{*\})\times R(X) \to R(X)$ for each set $X$, where $\s_X(t,u)$ is meant to model the explicit substitution $t[*/u]$.

\begin{definition}
The signature $\Upsilon_{\LCex}$ for the monad of the lambda calculus with explicit substitutions without equations is the coproduct of $\Theta'\times \Theta$ and $\Sigma_\LC$.
\end{definition}

\subsubsection{Equation}
The syntax of lambda calculus with explicit substitutions of \cite{lambda-ex}
is subject to the equation (see \cite[Figure 1, ``Equations'']{lambda-ex})
\begin{equation}
\label{eq:comm-subs}
t[x/u][y/v] = t[y/v][x/u]  
\quad
\text{if $y\notin \fv(u)$ and $x\notin\fv(v)$} \enspace .
\end{equation}
We rephrase it as an equality between two parallel $\Upsilon_{\LCex}$-module morphisms from
$\Theta''\times\Theta\times \Theta$, modeling the metavariables $t$, $u$, and
$v$, to $\Theta$:
\begin{equation}
\label{eq:subst-explicit}
  \xymatrix@C=25pt@R=2pt{
    \Theta'' \times \Theta\times \Theta \ar[rr]^{\Theta''\times \iota \times \Theta} &&
    \Theta'' \times \Theta' \times \Theta \ar[rr]^{\s'\times \Theta} & &
    \Theta'\times \Theta \ar[r]^-{\s} &
    \Theta
    \\
    \Theta'' \times \Theta\times\Theta
    \ar[rr]_{\Theta''\times \Theta\times \iota}
    &&
    \Theta'' \times \Theta\times\Theta'
    \ar[rr]_{\langle  \dswap\s\circ\pi_{1,3} , \pi_2\rangle}
    &&
    \Theta'\times \Theta \ar[r]_-{\s} &
    \Theta
    }
\end{equation}
Here, $\iota$ denotes the canonical morphism $\Theta \to \Theta'$ as before, and $\dswap\esubst$ is
the composition of $\esubst'$ with $\swap : \Theta'' \to \Theta''$ swapping the
two fresh variables.

Now we are ready to define the signature of the lambda calculus monad with explicit substitutions:
\begin{definition}
\label{d:signature-monad-lcex}
The signature $\Sigma_\LCex$ of the lambda calculus monad with explicit
substitutions  consists of $\Upsilon_\LCex$
and the single $\Sigma_\LCex$-equation stating the equality between the two
morphisms of Equation~\ref{eq:subst-explicit}.
\end{definition}

\begin{lemma}
\label{l:signature-monad-lcex-effective}
The signature $\Sigma_\LCex$ for monads is effective
\end{lemma}
\begin{proof}
 This is a direct consequence of Corollary~\ref{cor:effective-if-2-sig-alg}.
 \end{proof}

\subsection{Reduction Rules for Lambda Calculus with Explicit Substitutions}
\label{ss:reduction-rules-lambda-ex}
The reduction signature for the lambda calculus with explicit substitutions consists of two components: the first one is the signature for monads $\Sigma_\LCex$ of Definition~\ref{d:signature-monad-lcex};
the second one is the list of reduction rules that we enumerate here, taken from \cite[Figure 1, ``Rules'']{lambda-ex}.
Except for congruence, none of them involve hypotheses.

First, let us state the congruence rules (that are implicit in \cite{lambda-ex}):
\[
\frac{
\redrule{T}{T'} 
}
{\redrule{\app(T,U)}{\app(T',U)} }
\quad
\frac{
\redrule{U}{U'} 
}
{\redrule{\app(T,U)}{\app(T,U')} }
\quad
\frac{
\redrule{T}{T'} 
}
{\redrule{\abs(T)}{\abs(T')} }
\]
\[
\frac{
\redrule{T}{T'} 
}
{\redrule{\esubst(T,U)}{\esubst(T',U)} }
\quad
\frac{
\redrule{U}{U'} 
}
{\redrule{\esubst(T,U)}{\esubst(T,U')} }
\]
They are translated into reduction rules through the protocol described in Section~\ref{ss:protocole-reduction-rules}.

\begin{figure}
\[
  \AxiomC{}
  \RightLabel{\scriptsize{$\beta$-red}}
  \UnaryInfC{$\redrule{(\lambda x.t)u}{t[x/u]}$}
   \DisplayProof 
\quad
\AxiomC{$x\notin \fv(t)$}
\RightLabel{\scriptsize{Gc}}
\UnaryInfC{$\redrule{t[x/u]}{t}$ }
\DisplayProof
   \quad
   \AxiomC{}
   \RightLabel{\scriptsize{var$[]$}}
   \UnaryInfC{$\redrule{x[x/u]}{u}$}
   \DisplayProof 
    \]
     \bigskip
    \[
      \AxiomC{}
      \RightLabel{\scriptsize{$\app[]$}}
      \UnaryInfC{$\redrule{(t\,u)[x/v]}{t[x/v]\,u[x/v]}$}
      \DisplayProof
      \quad
      \AxiomC{}
      \RightLabel{\scriptsize{$\abs[]$}}
      \UnaryInfC{$\redrule{(\lambda y.t)[x/v]}{\lambda y.t[x/v]}$}
\DisplayProof
\]
\bigskip
\[
\AxiomC{$x\notin\fv(v)$}
\AxiomC{$y\in\fv(u)$}
\RightLabel{\scriptsize{$[][]$}}
\BinaryInfC{$\redrule{t[x/u][y/v]}{t[y/v][x/u[y/v]]}$}
\DisplayProof
    \]
\caption{Reduction rules of lambda calculus with explicit substitutions.}
\label{fig:lambdaex-reds}
\Description{The figure specifies reduction rules of lambda calculus with explicit substitutions à la Kesner.}
\end{figure}

\begin{figure}
  \begin{equation*}
    \AxiomC{}
    \RightLabel{\scriptsize{$\beta$-red}}
    \UnaryInfC{$\redrule{\app(\abs(T),U)}{\esubst(T,U)}$}
    \DisplayProof
    \quad
    %
    \AxiomC{}
    \RightLabel{\scriptsize{Gc}}
    \UnaryInfC{$\redrule{\esubst(\iota(T),U)}{T}$}
    \DisplayProof
    \end{equation*}
    \bigskip
    \begin{equation*}
    \AxiomC{}
    \RightLabel{\scriptsize{$\app[]$}}
    \UnaryInfC{$\redrule{\esubst(\app(T,U),V)}{\app(\esubst(T,V), \esubst(U,V) )}$}
    \DisplayProof
    \quad
  \end{equation*}
  \bigskip
  \begin{equation*}
    \AxiomC{}
    \RightLabel{\scriptsize{$\var[]$}}
    \UnaryInfC{$\redrule{\esubst(*,T)}{T}$}
    \DisplayProof
    \quad
    %
    \AxiomC{}
    \RightLabel{\scriptsize{$\abs[]$}}
    \UnaryInfC{$\redrule{\esubst(\abs'(T),V)}{\abs(\dswap{\esubst}(T,\iota(V)))}$}
    \DisplayProof
  \end{equation*}
  \bigskip
  \begin{equation*}
    \AxiomC{}
    \RightLabel{\scriptsize{$[][]$}}
    \UnaryInfC{$\redrule{\esubst(\esubst'(T,\kappa(U)),V)}{\esubst(\dswap{\esubst}(T,\iota(V)),\esubst(\kappa(U),V))}$}
    \DisplayProof
  \end{equation*}

  \begin{align*}
    \kappa &: \Theta_* \to \Theta' \\
      \dswap{\s} & : \Theta'' \times \Theta' \to \Theta'
    \end{align*}
     \flushleft
    where $\Theta_*$ is the 1-hole context $\Sigma_{\LCex}$-submodule of
    $\Theta'$ (Definition~\ref{d:one-hole-context-submodule}), and
    $\dswap{\s}$ 
      is defined as the composition
      \[
      \xymatrix{
      \Theta'' \times\Theta' \ar[r]^{\swap\times \Theta'} &
                                      \Theta'' \times \Theta' \ar[r]^-\s &
                                      \Theta'
                                      }
      \]
      Here, $\swap$ exchanges the fresh variables:
\begin{align*}
          \swap_{X,R} & : R((X+\{*_1\}) + \{ *_2\}) \to R((X+\{*_1\}) + \{ *_2\})(t) \\
    \swap_{X,R} & : t \mapsto t\monsubst{*_1 := *_2 ; *_2 := *_1}
                  \end{align*}
\caption{Reduction rules of Figure~\ref{fig:lambdaex-reds} reformulated in our setting.} 
\label{fig:reds-lambdaex-details}
\Description{The figure specifies reduction rules of lambda calculus with explicit substitutions à la Kesner, formulated in our framework.}
\end{figure}

Figure~\ref{fig:lambdaex-reds} gives Kesner's rules.
Five out of six of Kesner's rules translate straightforwardly, see Figure~\ref{fig:reds-lambdaex-details}.
Note how the explicit weakening $\iota : \Theta\to\Theta'$ accounts for
the side condition $x\notin\fv(t)$ of the Gc-rule in Figure~\ref{fig:lambdaex-reds}.

Expressing the side condition $y \in \fv(u)$ of the [][]-rule of Figure~\ref{fig:lambdaex-reds} requires
the definition of the $\Sigma_\LCex$-module $\Theta_*$ such that $\LCex_*$ is
the submodule of $\LCex'$ of terms that really depend on the fresh variable.
 We propose an approach based on the informal intuitive idea of defining
 inductively the submodule $R_*$ of elements in $R'$ having at least one occurrence of the fresh variable $*$ as
 follows, for a given model $R$ of $\Sigma_\LCex$:
   \begin{itemize}
   \item $\eta(*)\in R_*(X)$, for any set $X$;
   \item (application)
     \begin{itemize}
       \item if $t\in R(X)$ and $u\in R_*(X)$, then $\app(\iota (t),u)\in R_*(X)$ 
       \item if $t\in R_*(X)$ and $u\in R(X)$, then $\app(t,\iota(u))\in R_*(X)$ 
       \item if $t\in R_*(X)$ and $u\in R_*(X)$, then $\app(t,u)\in R_*(X)$ 
     \end{itemize}
   \item if $t \in R_*(X+\{ x \})$,
        then $\lambda x.t \in R_*(X)$;
   \item (explicit substitution)
     \begin{itemize}
       \item if $t\in R(X+\{ x\})$ and $u\in R_*(X)$, then $\iota(t)[x/u]\in R_*(X)$;
       \item if $t\in R_*(X+\{x\})$ and $u\in R(X)$, then $t[x/\iota(u)]\in R_*(X)$;
       \item if $t\in R_*(X+\{x\})$ and $u\in R_*(X)$, then $t[x/u]\in R_*(X)$.
     \end{itemize}
       \end{itemize}
       Guided by this intuition, we now formally define 
  a $\Sigma_\LCex$-module $\Theta_*$ equipped with a morphism $\kappa:\Theta_*\to\Theta'$.

    The previous informal inductive
    definition is translated as an initial algebra for an endofunctor
    on the category of $\Sigma_\LCex$-modules, which is cocomplete (colimits
    are computed pointwise). This endofunctor maps a $\Sigma_\LCex$-module $M$
    to the coproduct of the following $\Sigma_\LCex$-modules:
    \begin{itemize}
    \item the terminal $\Sigma_\LCex$-module $1$, playing the rôle of the fresh variable;
      \item the coproduct $M\times \Theta + \Theta \times M + M \times M$, one summand
        for each case of the application;
        \item the derived module $M'$ for abstraction;
      \item the coproduct $M'\times \Theta  + \Theta'\times M +
  M'\times M$, one summand
        for each case of the explicit substitution.
      \end{itemize}
  This functor is finitary, so the initial algebra exists thanks to Ad\'amek's
  theorem (already cited, as Theorem~\ref{l:adamek}).
      Unfortunately,
      the resulting $\Sigma_{\LCex}$-module does not yield
      the module that we are expecting in the case of the monad $\LCex$: it does
      not satisfy Equation~\ref{eq:subst-explicit}, and thus contains more terms
      than necessary.
      To obtain the desired $\Sigma_{\LCex}$-module, we 
       equip $\Theta'$ with its
      canonical algebra structure, inducing a morphism from the initial algebra,
      and we define $\Theta_*$ as the image of this morphism, thus equipped with
      an inclusion $\kappa : \Theta_* \to \Theta'$.

      \begin{definition}
        \label{d:one-hole-context-submodule}
        We define  the \textbf{$\Sigma_\LCex$-module of ``one-hole contexts''} to be $\Theta_*$, equipped with an inclusion $\kappa : \Theta_*\to\Theta'$.
        \end{definition}
        \begin{rem}
          Such a definition can be worked out for any algebraic signature for monads.
          \end{rem}

\noindent
Now we define the signature of the reduction monad of lambda calculus with explicit substitutions:
\begin{definition}
\label{d:sig-red-monad-lex}
The reduction signature $\mathcal{S}_\LCex$ of the lambda calculus reduction monad with explicit substitutions consists of the signature $\Sigma_\LCex$ of Definition~\ref{d:signature-monad-lcex} and all the reduction rules specified in this section.
\end{definition}
\begin{lemma}
\label{l:sig-red-monad-lex-eff}
The reduction signature $\mathcal{S}_\LCex$ is effective.
\end{lemma}
\begin{proof}
Apply Theorem~\ref{thm:alg-2-sig-initial}. The underlying signature for monads is effective by Lemma~\ref{l:signature-monad-lcex-effective}.
\end{proof}

\section{Recursion}
\label{sec:recursion}
In this section, we derive, for any effective reduction signature $\mathcal{S}$,
a recursion principle from initiality.
In Section~\ref{ss:rec}, we state this recursion principle, then we give an
example of application in Section~\ref{ss:fixpoint-to-beta}, by translating lambda
calculus with a fixpoint operator to lambda calculus.
In Section~\ref{ss:fixpoint-to-beta}, we apply this principle to translate
lambda calculus with explicit substitutions into lambda calculus with 
\emph{unary congruent  substitution}.
Then, in Section~\ref{ss:induction-principle}, we translate this latter variant
of lambda calculus into lambda calculus closed under identity and
composition of reductions.

\subsection{Recursion Principle for Effective Signatures}
\label{ss:rec}

The recursion principle associated to an effective signature provides a way to construct 
a morphism
from the reduction monad underlying the initial model
of that signature to a given reduction monad.

\begin{proposition}[Recursion principle]
Let $\mathcal{S}$ be an effective reduction signature, and
   $R$ be the reduction monad underlying the initial model. Let $T$
  be a reduction monad. Any action $\tau$ of $\mathcal{S}$ in $T$ induces a reduction monad morphism $\hat \tau : R \to T$.  
\end{proposition}
\begin{proof}
The action $\tau$ defines a model $M$ of $\mathcal{S}$. By initiality, there is a
unique model morphism from the initial model to $M$, and $\hat \tau$ is the
reduction monad morphism underlying it.
\end{proof}

In the next sections, we illustrate this principle. 

\subsection{Translation of lambda calculus with fixpoint to lambda calculus}
\label{ss:fixpoint-to-beta}

In this section, we consider the signature
 $\mathcal{S}_{\LCfix} $ of Example~\ref{ex:redsig-lc-fixpoint} for the lambda calculus
 with an explicit fixpoint operator.

 We build, by recursion, a reduction monad morphism from the initial model
$\LCfix$ of this signature to $\LCclot$, the \enquote{closure under identity
  and composition of reductions} (Definition~\ref{d:transitive-closure-redmonad})
of the initial model
 $\LCb$ of the signature $\mathcal{S}_{\LCb}$ (Example~\ref{ex:variants-lc}).

As explained in Section~\ref{ss:rec}, we need to define an action of
$\mathcal{S}_{\LCfix}$ in  $\LCclot$.
Note that $\mathcal{S}_\LCfix$ is an extension of $\mathcal{S}_{\LCb}$
(Example~\ref{ex:redsig-lc-fixpoint}). 
First, we focus on the core $\mathcal{S}_{{\LCb}}$ part:
 we show that the reduction monad $\LCclot$ inherits the canonical action of
 $\mathcal{S}_{\LCb}$ in $\LCb$.
\begin{lemma}
  \label{l:action-lcbered-star}
  There is an action of $\mathcal{S}_{\LCb}$ in $\LCclot$.
\end{lemma}

We have formalized a proof of this statement in Agda.%
\footnote{
  The source code can be downloaded from
  \url{https://github.com/amblafont/unary-subst-LCstar/blob/master/fiberlambda.agda.}}

\begin{proof}
  The challenge is to give an action of reduction rules with hypotheses:
  now the input reductions of the rule may be actually sequences of reductions.
  This concerns congruence for application and abstraction.
  We take the example of abstraction: suppose we have a sequence of reductions
  $r_1\dots r_n$ going from $t_0$ to $t_n$. We want to provide a reduction between
  $\abs(t_0)$ and $\abs(t_n)$. 
  For each $i$, we have a reduction between $\abs(t_{i-1})$ and $\abs(t_i)$.
  By composing the corresponding sequence, we obtain the desired reduction.
\end{proof}

 The action for the extra parts of $\mathcal{S}_\LCfix$ requires the following:
\begin{itemize}
  \item An operation $\fix : {\LCclot}' \to \LCclot$.
    A fixpoint combinator $Y$ is a closed term with the property that for any other term $t$, the
term $\app(Y,t)$ $\beta$-reduces in some steps to $\app(t,\app(Y,t))$.
    Here, we choose a fixpoint combinator $Y$ (for example, the one of Curry), and set
    $\fix_X(t)=\app(Y,\abs(t))$, in accordance with \cite[Section 8.4]{ahrens_et_al:LIPIcs:2018:9671}.
 \item An action of the reduction rule
\begin{prooftree}
  \AxiomC{}
  \UnaryInfC{$\redrule{\fix(T)}{T\monsubst{*:=\fix(T)}}$}
\end{prooftree}
We denote by $r \in \Redof{\LCclot}(\{*\})$ a reduction between
$\app(Y,*)$ and $\app(*,\app(Y,*))$. Then, $r$ induces an $\LC$-module
morphism $\hat{r} : \LC\to \Redof{\LCclot}$ by mapping an element $t\in\LC(X)$
to $r\monsubst{*:=t}$.
We define the action of this reduction rule as the composition of the following reductions:
\[
  \app(Y,\abs(t)) \rightsquigarrow_{\hat{r}(\abs(t))} \app(\abs(t),\app(Y,\abs(t)))
      \rightsquigarrow_{\beta} t\monsubst{*:= \app(Y,\abs(t))}
\]
    \item an action of the congruence rule
        \[
          \frac{
           \redrule{T}{T'}
          }
          {
            \redrule{\fix(T)}{\fix(T')}
          }
        \]
        This can be defined in the obvious way using the congruences of application and abstraction.
  \end{itemize}

In more concrete terms, our translation is a kind of compilation which replaces each occurrence of the explicit fixpoint operator $\fix(t)$
with $\app(Y,\abs ( t))$, and each fixpoint reduction with a composite of $\beta$-reductions.

\subsection{Translation of Lambda Calculus with Explicit Substitutions into Lambda
  Calculus with Congruent Unary Substitution} 
\label{sec:translate-rec}

Here, we consider the reduction signature $\mathcal{S}_\LCex = (\Sigma_\LCex,
\mathfrak{R}_\LCex)$ introduced in Definition~\ref{d:sig-red-monad-lex}. 
 The underlying monad of the initial model $\LCex$ is the monad 
of lambda calculus with an application and abstraction operation, and an
explicit substitution operator $\LCex'\times\LCex\to \LCex$ satisfying Equation~\ref{eq:subst-explicit}, for $R=\LCex$.
The associated reduction monad has all the rules specified in Section~\ref{s:ex-lex}.
 
We build, by recursion, a reduction monad morphism from the initial model
$\LCex$ of this signature to $\LCcong$, a variant of the lambda calculus
specified by the signature $\mathcal{S}_{\LCb}$ (Example~\ref{ex:variants-lc}) extended with
the \emph{congruence for unary  substitution}:
\[
  \frac{\redrule{T}{T'}}{\redrule{U\monsubst{*:=T}}{U\monsubst{*:=T'}}}
\]
Note that this reduction rule accounts for the reflexivity rule, and 
makes congruences for application (but not congruence for abstraction) redundant:
\begin{description}
\item[Reflexivity] Any $U\in \LC(X)$ can be weakened into $\iota(U)\in\LC'(X)$.
  Then, consider any reduction $m:\redfib{T}{T'}$ . The action of
  the reduction rule above yields a reduction between $\iota(U)\monsubst{*:=T} =
  U$ and $\iota(U)\monsubst{*:=T'}=U$. By choosing $m$ adequately  
  (for example, take the $\beta$-reduction between $\app(\abs(*),\abs(*))$ and
  $\abs(*)$), this yields an action of the reflexivity reduction rule.
  \item[Congruence] Consider the left congruence rule (the cases of the right one
    and the congruence for abstraction are similar): from any reduction
    $m:\redfib{T}{T'}$, we want a reduction between $\app(T,U)$ and
    $\app(T',U)$, for $T,T',U\in \LC(X)$. We obtain it by applying the action of
    unary congruent substitution to $m$ for the term $\app(*,U)$. One checks
    that this indeed defines an action for the left congruence reduction rule of application.
\end{description}

\noindent
As explained in Section~\ref{ss:rec}, we need to define an action of $\mathcal{S}_{\LCex}$ in  $\LCcong$:
\begin{itemize}
\item the operations of application and abstraction are those of $\LCcong$ as
  the initial model of $\Sigma_\LC$ (recall that the underlying monad of $\LCcong$
  is just $\LC$);
  \item the explicit substitution operation ${\LCcong}'\times\LCcong \to \LCcong$ is defined
using the monadic substitution, mapping a pair $(t,u)\in\LCcong(X+\{*\})\times\LCcong(X)$ to the monadic substitution $t\monsubst{*:= u}$;
\item Equation~\ref{eq:subst-explicit} for the underlying monad is satisfied thanks to the usual
monadic equations;
\item the action of the congruence rules for application and abstraction are
  induced by the action of the congruence rule for unary substitution, as
  explained above;
  \item $\LCcong$ has already an action for $\beta$-reduction;
  \item all the actions for the remaining reduction rules involving explicit
    substitution (except the congruences for explicit substitution that are
    discussed below) are given by an action of the reflexivity reduction rule;
    \item
 the non-obvious actions are the ones of the congruence rules for explicit substitution:
\begin{align*}
     \frac{  \redrule{T}{T'}}
    {  \redrule{T[x/U]}{T'[x/U]} }
    \qquad
         \frac{\redrule{U}{U'}}
    {  \redrule{T[x/U]}{T[x/U']} }
    \end{align*}
    The left one is obtained from the substitution of the module of reductions (see Remark~\ref{r:congr-weak-subst}).  
    The right one is exactly given
    by the action of the congruence rule for unary substitution.
    \end{itemize}

Finally, by the recursion principle, we get a reduction monad morphism from $\LCex$ to $\LCcong$.
This translation replaces the explicit substitution operator $t[x/u]$ with the
corresponding monadic substitution $t\monsubst{x:=u}$, and all the reductions 
are translated to 
reflexivity except for the ones for the $\beta$-reduction and congruences.

\subsection{Translation of Lambda Calculus with Congruent Unary 
  Substitution into Lambda Calculus}
\label{ss:induction-principle}
In the previous section, we translated lambda calculus with explicit
substitution into lambda calculus with congruent unary  substitution. In
this section, we translate this variant of lambda calculus into $\LCclot$
(introduced in Section~\ref{ss:fixpoint-to-beta}),  
the closure under identity and composition of reductions 
the initial model
 $\LCb$ of the signature $\mathcal{S}_{\LCb}$ (Example~\ref{ex:variants-lc}).

As per Section~\ref{ss:rec}, we need to define an action in $\LCclot$ of
the signature $\mathcal{S}_{\LCb}$ extended with the reduction rule:
\begin{equation}
  \frac{\redrule{T}{T'}}{\redrule{U\monsubst{*:=T}}{U\monsubst{*:=T'}}}
  \label{eq:unary-subst-cong}
\end{equation}
Thanks to Lemma~\ref{l:action-lcbered-star}, we have an action of
$\mathcal{S}_{\LCb}$ in $\LCclot$.
Thus, the main challenge consists in equipping $\LCclot$ with an action of the rule~\eqref{eq:unary-subst-cong}.
   \begin{proposition}
     The reduction monad $\LCclot$ can be equipped
     with an action of the 
      rule~\eqref{eq:unary-subst-cong}.
    \end{proposition}
We have formalized a proof of this statement in Agda.%
\footnote{
  The source code can be downloaded from
  \url{https://github.com/amblafont/unary-subst-LCstar/blob/master/fiberlambda.agda.}}
    \begin{proof}
      We will write $r:R\to\LC\times\LC$ in place of $\redof{\LCclot}:\Redof{\LCclot}\to\LC\times\LC$.

    Such  an action is equivalently given (see Lemma~\ref{l:def-alt-action}) by
     a morphism $\alpha :
    \LC'\times R\to R$ such that the following
    diagram commutes, where $q_X(t,m) = \bigl(t\monsubst{*:=\source{}(m)}, t\monsubst{*:=\target{}(m)}\bigr)$.
    \begin{equation}
      \label{eq:action-unary-subst-commute}
      \xymatrix@R=0.7pc{
        \LC'\times R\ar@{-->}[rr]^\alpha
        \ar[rd]_{q}
        &&  R
       \ar[ld]^{r}
       \\
       & \LC\times\LC
      }
    \end{equation}

    We first construct the collection of functions $(\alpha_X)_X$ with $\alpha_X:\LC'(X)\times
    R(X)\to R(X)$ and then show the two required properties, i.e., that it
    commutes with substitution (thus inducing a $\LC$-module morphism), and that
    it satisfies Equation~\ref{eq:action-unary-subst-commute}.

    The construction of the collection of functions (without naturality conditions) is worked
    out in the functor category $[\Set_0,\Set]$, where $\Set_0$ is the
    discretized category of sets (this base category allows us to get rid of
    naturality conditions). This is done by recursion on the first argument.
 More formally, we exploit some initiality
property of
$\LC'\cdot j$, where $j:\Set_0\to\Set$ is the inclusion of the discretized category
of sets into sets. Indeed, $\LC'\cdot j$ is the initial algebra of the
endofunctor
\begin{align*}
  \Psi:[\Set_0,\Set]  \to [\Set_0,\Set]
  \qquad
  F  \mapsto  j+1+F\times F + F'
  \end{align*}
where $F'$ is the functor mapping a set $X$ to $F(X+1)$.

The two properties that we want to show about the collection of functions $(\alpha_X)_X$ are
then proved by \emph{induction} on the first argument, again exploiting initiality
of $\LC'\cdot j$, as we explain below.
The proof goes as follows:
\begin{enumerate}
\item construct (by initiality) a morphism from $\LC'\cdot j$ to
  the exponential of $R\cdot j$ with itself, that is, to the functor
    $R^{R} :  \Set_0 \to\Set$ defined on objects by
    $X  \mapsto R(X)^{R(X)}$;
  \item show that the induced morphism from $\LC'\cdot j \times
    R\cdot j \to R\cdot j$ yields a $\LC$-module
    morphism $\alpha : \LC'\times R\to R$;
  \item show the commutation required by
    Equation~\ref{eq:action-unary-subst-commute}.
\end{enumerate}
Note how working in the functor category $[\Set_0,\Set]$ allows us to define the
functor $R^{R}$ as
above, without worrying about the functorial action on morphisms.
Below we  sometimes omit the explicit precomposition with $j$
in order to simplify the notation.
Now we perform the steps listed above.
\begin{enumerate}
\item As we argued before, $\LC'\cdot j$ is the initial algebra of $\Psi$, so
  our task consists in equipping $R^{R}$ with an
  algebra structure for $\Psi$, that we split into the following four components,
  using the universal properties of the coproduct and the exponential in the
  category $[\Set_0,\Set]$:
  \begin{enumerate}
  \item The morphism $j\times R  \to R$ corresponds to
    the case of variables.
We expect that the resulting module morphism $\alpha : \LC'\times R \to R$ satisfies
    $\alpha_X(\eta(x), m)=\refl_X(x)$ for any $x\in X$,
    where
    $\refl:\LC\to R$ maps a term to the reflexive reduction on itself
   and $\eta:\Id\to\LC'$ is the unit of the monad $\LC$.
    Accordingly, we define the morphism $j\times R\to R$ as mapping a pair $(x,m)\in X\times
    R(X)$ to $\refl_X(\eta(x))$.
  \item
    The morphism $R\to R$ corresponds to the case of
    the fresh variable $*$. We expect that $\alpha_X(*,m)=m$. Accordingly, the
    required morphism is taken as the identity on $R$.
    \item The morphism
      $(R\times R)^{R} \to R^{R}$ 
      corresponds to the case of an application. We expect that
      $\alpha_X(\app(t,u),m)=\appcongr(\alpha_X(t,m),\alpha_X(u,m))$, where
      $\appcongr:R\times R \to R$ is
      the action of the reduction rule of congruence for application defined as
      $\appcongr(m_1,m_2)=\trans(\appcongr_1(m_1),\appcongr_2(m_2))$, where
      $\trans$ denotes an action of the transitivity reduction rule with which we can
      equip $\LCclot$ (by concatenating sequences of reductions).
      Accordingly, the morphism is defined as $\appcongr^{R}$.
      \item The morphism $R'^{R'}\to R^R$ corresponds to the case of an
        abstraction. We expect that
        $\alpha_X(\abs(t),m)=\abscongr_X(\alpha_{X+1}(t,R\iota_X(m)))$, where
        $\iota : \Id \to \Id'$ is the canonical inclusion.
        Accordingly, we take $\abscongr^{R\iota}$ as the the required morphism.
  \end{enumerate}
  By initiality, we get an algebra morphism from $\LC\cdot j$ to $R^R$, which by
  uncurrying yields a morphism $\alpha:\LC'\cdot j \times R\to R$.
  \item Upgrading $\alpha$ into a module morphism from $\LC'\times R$ to $R$ consists
    in showing compatibility with substitution in the following sense: for any
    map $f : X \to \LC(Y)$, for any pair $(t,m)\in\LC'(X)\times R(X)$, the
    equality 
$
  \alpha_X(t,m)\monsubst{f} = \alpha_Y(t\monsubst{f}, m\monsubst{f})
$
is satisfied.
This is shown by induction on $t\in\LC'(X)$. The case of variables
requires a preliminary step: for $t=\eta(x)$, the equation amounts to
$\refl(f(x))=\alpha(\LC i(f(x)),m)$, which is not straightforward.
We hence first prove by induction on $t\in\LC(X)$ that $\alpha(\LC i(t),m)=\refl(t)$.
We do not detail these straightforward inductions, but rather
explain the general methodology to perform induction on $\LC'$ (the case of
$\LC$ is similar). Suppose given, for each $t\in \LC'(X)$, a predicate
$P_X(t)$. Then, one can form the functor $\LC'_{|P}:\Set_0\to\Set$ mapping a set
$X$ to the subset of $\LC'(X)$ satisfying the predicate $P_X$. It follows that  $\LC'_{|P}$
embeds into $\LC$, and if $\LC'_{|P}$ inherits the algebra structure for $\Psi$
through this embedding, then by initiality we get a section of the embedding,
which exactly translates the fact that any term $t\in\LC'(X)$ satisfies the property.

\item It remains to show the commutation of Diagram~\ref{eq:action-unary-subst-commute}.
  Again, an induction on the first argument (thus exploiting initiality of
  $\LC'\cdot j$) is enough to conclude.
  \qedhere
\end{enumerate}
\end{proof}

\section{Conclusion and future work}

We introduced the notions of reduction monad and 
reduction signature. 
For each such signature, we defined a category of models, equipped with a
forgetful functor to the category of reduction monads.
We say that a reduction signature is effective if its category of models has an initial object;
then, we say that the reduction monad underlying the initial object is generated by the signature.
Our main result identifies a simple sufficient condition for a reduction signature to be effective.

This work is the first step towards a theory for the algebraic specification of programming languages and their semantics. 
Future work could include
\begin{itemize}
\item generalizing our notion of signature to encompass richer languages;
\item extending our work to simply-typed languages;
\item proving modularity results for our signatures and their models, analogous
    to that of \cite[Theorem~27]{FSCD2019}.
\end{itemize}
Part of this is already under way~\cite{lafont2019signatures,HHL}.

\begin{acks}
We thank the referees for their careful reading and thoughtful and constructive criticism.
Furthermore, we thank Tom Hirschowitz for many discussions about this work and, in particular, for sharing his knowledge about related work.

Ahrens acknowledges the support of the \grantsponsor{}{Centre for Advanced Study (CAS)}{} in Oslo, Norway, which funded and hosted the research project \emph{\grantnum{}{Homotopy Type Theory and Univalent Foundations}} during the 2018/19 academic year.

Lafont has been supported by the \grantsponsor{}{European Research Council}{} under Grant No.~\grantnum{CoqHoTT}{637339}.

Maggesi has been supported by \grantsponsor{GNSAGA-INdAM}{Gruppo Nazionale per le Strutture Algebriche, Geometriche e le loro Applicazioni (GNSAGA-INdAM)}{https://www.altamatematica.it/gnsaga/} and \grantsponsor{MIUR}{Ministero dell'Istruzione, dell'Università e della Ricerca (MIUR)}{https://www.miur.gov.it}.
\end{acks}

\bibliography{strengthened}

\end{document}